%% file: IPDPS13-DEX.tex
\documentclass[11pt]{article}

\usepackage{times}
\usepackage{threeparttable}
\usepackage{fullpage}

\usepackage{etoolbox}
\makeatletter
\patchcmd{\maketitle}{\@copyrightspace}{}{}{}
\makeatother

\usepackage{calc}
\usepackage{paralist}
\usepackage{amsmath,amsfonts,amssymb}
\usepackage{color,graphicx,subfigure,boxedminipage}
\usepackage{hyperref,verbatim}
\usepackage[noend]{algorithmic}
\usepackage[section,boxed]{algorithm}
\usepackage{setspace}
\usepackage{ifthen}

\newcommand{\squishlist}{
 \begin{list}{$\bullet$}
  { \setlength{\itemsep}{0pt}
     \setlength{\parsep}{3pt}
     \setlength{\topsep}{3pt}
     \setlength{\partopsep}{0pt}
     \setlength{\leftmargin}{1.5em}
     \setlength{\labelwidth}{1em}
     \setlength{\labelsep}{0.5em} } }

\newcommand{\squishlisttwo}{
 \begin{list}{$\bullet$}
  { \setlength{\itemsep}{0pt}
     \setlength{\parsep}{0pt}
    \setlength{\topsep}{0pt}
    \setlength{\partopsep}{0pt}
    \setlength{\leftmargin}{2em}
    \setlength{\labelwidth}{1.5em}
    \setlength{\labelsep}{0.5em} } }

\newcommand{\squishend}{
  \end{list}  }

\newboolean{short}
\setboolean{short}{false}  
 
\newcommand{\onlyShort}[1]{\ifthenelse{\boolean{short}}{#1}{}}
\newcommand{\onlyLong}[1]{\ifthenelse{\boolean{short}}{}{#1}}

\newcommand{\dex}{\textsc{dex}}

\newcommand{\Low}{\textsc{Low}}
\newcommand{\High}{\textsc{Spare}}
\newcommand{\Spare}{\textsc{Spare}}
\newcommand{\cloud}[1]{\textsc{cloud}\ensuremath{(#1)}}

\newcommand{\insertion}{\texttt{insertion}}
\newcommand{\deletion}{\texttt{deletion}}
\newcommand{\computeHigh}{\texttt{computeSpare}}
\newcommand{\computeLow}{\texttt{computeLow}}

\newcommand{\inflate}{\texttt{simplifiedInfl}}
\newcommand{\lightInflate}{\texttt{inflate}}
\newcommand{\stagInflate}{\texttt{inflate}}
\newcommand{\deflate}{\texttt{simplifiedDefl}}
\newcommand{\lightDeflate}{\texttt{deflate}}
\newcommand{\stagDeflate}{\texttt{deflate}}
\newcommand{\graph}[1]{\langle#1\rangle}
\newcommand{\ra}{\rightarrow}

\newcommand{\eps}{\varepsilon}
\newcommand{\Sim}{\textsc{Sim}}
\newcommand{\NewSim}{\textsc{NewSim}}
\newcommand{\Load}{\textsc{Load}}
\newcommand{\NewLoad}{\textsc{NewLoad}}
\newcommand{\Z}{\mathcal{Z}}

\newcommand{\G}{\mathcal{G}}
\newcommand{\red}[1]{{\textcolor{red}{#1}}}
\newcommand{\todo}[1]{\noindent\textbf{TODO: }\marginpar{****}%
\renewcommand{\le}{\leqslant}
\renewcommand{\ge}{\geqslant}
\renewcommand{\leq}{\leqslant}
\renewcommand{\geq}{\geqslant}
\textit{\red{{#1}}}\textbf{ :ODOT}}

\newtheorem{definition}{Definition}
\newtheorem{lemma}{Lemma}
\newtheorem{claim}{Claim}
\newtheorem{corollary}{Corollary}
\newtheorem{fact}{Fact}
\newtheorem{theorem}{Theorem}

\DeclareGraphicsExtensions{.pdf}
\graphicspath{{images/}}

\newcommand{\E}{\mathrm{E}}

\newcommand{\Peter}[1]{\{\textbf{Peter:} #1 \}}

\newcommand{\dist}{\mathrm{dist}}
 
\makeatletter
\renewcommand\section{\@startsection {section}{1}{\z@}%
                                  {-2ex \@plus -1ex \@minus -.2ex}%
                                  {2ex \@plus.2ex}%
                                  {\normalfont\normalsize\bfseries}}
\renewcommand\subsection{\@startsection{subsection}{2}{\z@}%
                                    {-2ex\@plus -1ex \@minus -.2ex}%
                                    {1ex \@plus .2ex}%
                                    {\normalfont\normalsize\bfseries}}
\makeatother
\newcommand{\sq}{\hbox{\rlap{$\sqcap$}$\sqcup$}}
\newcommand{\qed}{\hspace*{\fill}\sq}
\newenvironment{proof}{\noindent {\bf Proof.}\ }{\qed\par\vskip 4mm\par}

\begin{document}

\clubpenalty=10000 
\widowpenalty = 10000

\renewcommand{\le}{\leqslant}
\renewcommand{\ge}{\geqslant}
\renewcommand{\leq}{\leqslant}
\renewcommand{\geq}{\geqslant}

\newcommand{\Prob}[1]{\text{Pr}\left[#1\right]}



\title{DEX: Self-healing  Expanders}

\author{Gopal Pandurangan\thanks{ 
   Division of Mathematical Sciences, Nanyang Technological University, Singapore 637371
 and Department of Computer Science, Brown University, Providence, RI 02912, USA.   Email: {\tt gopalpandurangan@gmail.com}.
   Supported in part by Nanyang Technological University grant M58110000, Singapore Ministry of Education (MOE) Academic Research Fund (AcRF) Tier 2 grant MOE2010-T2-2-082, MOE  AcRF Tier 1 grant MOE2012-T1-001-094, and by a grant from the United States-Israel Binational Science Foundation (BSF).}
   \and
   Peter Robinson\thanks{  
 Division of Mathematical Sciences,
  Nanyang Technological University, Singapore 637371. 
Email: {\tt peter.robinson@ntu.edu.sg}.
Supported in part by Nanyang Technological University grant M58110000 and Singapore Ministry of Education (MOE) Academic Research Fund (AcRF) Tier 2 grant MOE2010-T2-2-082.}
   \and
   Amitabh Trehan\thanks{ 
   School of Engineering and Computer Science
   The Hebrew University of Jerusalem
   Jerusalem, Israel - 91904 
Email: {\tt amitabh.trehaan@gmail.com}. 
This research project was supported by the Israeli Centers of Research Excellence (I-CORE) program (Center No. 4/11). Work done in part while the author was at the Technion and supported by a Technion fellowship.}}

\date{}
\maketitle 


\input Abstract
\vspace{-0.2in}



\input Introduction
\input Model
\input Main

\input Conclusions
\bibliographystyle{plain}
\bibliography{selfheal,papers}

\onlyLong{%


\appendix


\section{Previous Results and Definitions} \label{app-sec:preliminaries}

For completeness, we restate some definitions and results from  literature
that we reference in the paper.

We use the notation $G = \graph{n,d,\lambda_G}$ to denote a $d$-regular 
graph $G$ of $n$ nodes where the second largest eigenvalue of the 
adjacency matrix is $\lambda_G$.

\begin{definition}[Expanders, spectral gap] \label{def:expanders}
Let $d$ be a constant and let
$\G=(\graph{n_0,d,\lambda_0},\graph{n_1,d,\lambda_1},\dots)$ be an infinite
sequence of graphs where $n_{i+1}> n_i$ for all $i\ge 0$. We say that 
\emph{$\G$ is an
expander family of degree $d$} if there is a constant $\lambda<1$ such that $\lambda_i\le
\lambda$, for all $i\ge 0$. Moreover, the individual graphs in $\G$ are called 
\emph{expanders with spectral gap $1-\lambda$}.
\end{definition}

\begin{lemma}[cf.\ Lemma~1.15 in \cite{chungbook}] \label{lem:contraction}
  If $H$ is formed by vertex contractions from a graph $G$, then $\lambda_H \le
  \lambda_G$.
\end{lemma}

\begin{lemma} \label{lem:randomWalk}
  Consider an expander network and suppose that every node initiates a
  random walk of length $\Theta(\log n)$ and only $1$ random walk token
  can be sent over an edge in each direction in a round.
  Then all random walks have completed with high probability after
  $O(\log^2 n)$ rounds.
\end{lemma}
\begin{proof}
  The result follows by instantiating Lemma~2.2 of \cite{drw1}, which
  shows that, if every node initiates $\eta$ random walks of length $\mu$,
  then all walks complete within $O(\frac{\eta\mu\log n}{\delta})$ rounds
  where $\delta$ is the minimum node degree. 
\end{proof}

\begin{corollary}[Corollary~7.7.3 in \cite{S1998:Universal-abbrv}]
  \label{cor:routing}
  In any bounded degree expander of $n$ nodes, $n$ packets, one per node,
  can be routed according to an arbitrary permutation in
  $O\left(\frac{\log n (\log\log n)^2}{\log\log\log n}\right)$ rounds.
\end{corollary}

\begin{lemma}[Mixing Lemma, cf.\ Lemma~2.5~\cite{Wigderson-exsurvey}]
  \label{lemma:expmixlemma}
  Let $G$ be a $d$-regular graph of $n$
  vertices and spectral gap $1-\lambda$.
  Then,  for all set of nodes $S,T \subseteq V(G)$, we have that
$\left| |E(S,T)| - \frac{d|S||T|}{n} \right| \le \lambda d \sqrt{|S||T|}.$
\end{lemma}

\begin{definition}[Edge Expansion, \cite{Wigderson-exsurvey}]
  \label{def:edgeexpansion}
Consider a graph $G$ of $n$ nodes and a set $S \subseteq V(G)$.
Let $E(S,\bar{S})$ be the set of edges between $S$ and $G\setminus S$.
The \emph{edge expansion of $G$} is defined as
$$
h(G) := \min\left\{\frac{|E(S,\bar{S})|}{|S|}: \text{$S \subseteq V(G)$ and $|S|
\le n/2$}\right\}.
$$
\end{definition}

\begin{theorem}[Cheeger Inequality, Theorem~2.6 in
  \cite{Wigderson-exsurvey}] \label{thm:cheeger}
  Let $G$ be an expander with spectral gap $1-\lambda$ and edge expansion
  $h(G)$.
  Then
  $$
  \frac{1-\lambda}{2} \le h(G) \le \sqrt{2(1-\lambda)}.
  $$
\end{theorem}
}

\onlyLong{
\input AlgorithmMain

\input Extensions

}

\end{document}

%% file: Abstract.tex
\begin{abstract}
We present a fully-distributed self-healing algorithm \dex \, that  maintains a  constant degree expander network  in a dynamic setting.  To the best of our knowledge, our algorithm provides  the first  efficient distributed construction of   expanders  --- whose expansion properties  hold {\em deterministically}  --- that works even under an all-powerful adaptive adversary that controls the dynamic changes to the network (the adversary has unlimited computational power and knowledge of the  entire network state, can decide which nodes join and leave and at what time, and knows the past random choices made by the algorithm).  Previous distributed expander constructions  typically provide only {\em probabilistic} guarantees on the network expansion  which {\em rapidly degrade} in a dynamic setting; in particular, the expansion properties can  degrade even more rapidly under {\em adversarial} insertions and deletions.

Our algorithm provides efficient maintenance  and incurs a low overhead per insertion/deletion by an adaptive adversary: only  $O(\log n)$  rounds and $O(\log n)$ messages are needed with high probability  ($n$ is the number of nodes currently in the network). The algorithm requires only a constant number of topology changes.  
Moreover, our algorithm allows for an efficient implementation and maintenance of a distributed hash table (DHT) on top of \dex\, with only a constant additional overhead.


Our results are a step towards implementing  efficient self-healing  networks that have \emph{guaranteed} properties (constant bounded degree and expansion)   despite  dynamic changes.
\end{abstract}



%% file: Introduction.tex
\section{Introduction}
Modern networks (peer-to-peer, mobile, ad-hoc,  Internet, social, etc.) are dynamic and increasingly resemble self-governed living entities with largely distributed control and coordination. In such a scenario, the network topology governs much of the functionality of the network.
In what topology should  such nodes (having limited resources and
bandwidth) connect so that  the network has effective communication
channels with low latency for all messages, has constant degree, is robust
to a limited number of failures, and nodes can quickly sample a random
node in the network (enabling many randomized protocols)? The well known
answer is that they should connect as a (constant degree) {\em expander}
(see e.g., \cite{DBLP:books/wi/AlonS92}). How should such a topology be constructed in a distributed fashion? 
The problem is especially challenging in a {\em dynamic} network, i.e., a network exhibiting churn with nodes and edges entering and leaving the system. 
Indeed, it is a fundamental problem to  scalably  build dynamic topologies that 
have the desirable properties of an expander graph (constant degree and expansion, regardless of the network size) in  a distributed manner such that the expander properties are {\em always} maintained despite continuous network changes. Hence it is of both theoretical and practical interest to maintain expanders dynamically
in an efficient manner.


Many previous works (e.g., \cite{PRU01, LS03-abbrv, mihail-p2p2006}) have addressed the above  problem, especially in the context of building dynamic  P2P (peer-to-peer) networks.
However, all these constructions 
provide only  {\em probabilistic} guarantees of the expansion properties that  {\em degrade rapidly} over a
series of  network changes (insertions and/or deletions of nodes/edges) ---
in the sense that expansion properties cannot be maintained ad infinitum   due to
their probabilistic nature\footnote{For example,  even if the network is guaranteed to be an
expander with high probability (w.h.p.), i.e.\ a probability of $1-1/n^c$, for some
constant $c$, in every step (e.g., as in the protocols of \cite{LS03-abbrv} and \cite{PRU01}), the
probability of violating the expansion bound  tends to $1$ after some
polynomial number of steps.} which  can be a major drawback in a  dynamic setting.
In fact, the expansion properties can  degrade even more rapidly under 
adversarial insertions and deletions (e.g., as in \cite{LS03-abbrv}).  Hence, in a dynamic setting, guaranteed
expander constructions are needed.
Furthermore, it is important that the network maintains its expander properties
(such as high conductance, robustness to failures, and fault-tolerant
multi-path routing) {\em efficiently} even under dynamic network changes.
This will be useful in efficiently building good overlay and P2P network topologies
with expansion guarantees that do not degrade with time, unlike the above approaches.

Self-healing is  a \emph{responsive} approach to
fault-tolerance, in the sense that it responds to an attack (or component
failure) by changing the topology of the network.  This  approach works
irrespective of the initial state of the network, and is thus orthogonal
and complementary to traditional  non-responsive techniques.
Self-healing assumes the network to be \emph{reconfigurable} (e.g. P2P, wireless
mesh, and ad-hoc networks), in the sense that changes to the
topology of the network can be made on the fly.  Our goal is to design an efficient distributed self-healing algorithm
that maintains an expander despite attacks from an adversary.

\noindent {\bf Our Model:}
We use the self-healing model which is similar 
 to the model introduced in  \cite{HayesFG-DCJournal-springerlink,Amitabh-2010-PhdThesis} and is briefly described here (the detailed model is described in Sec.~\ref{sec: model}).
We assume an adversary that repeatedly attacks the network. This adversary is adaptive and knows the network topology and our algorithm (and also previous insertions/deletions and all previous random choices), and it has the ability to delete arbitrary nodes from the  network or insert a new node in the system which it can connect to any subset of  nodes currently in the system.  
We also assume that the adversary  can only delete or insert a single node at a time step.\onlyShort{ (Our algorithm can be extended to handle multiple insertions/deletions ---cf.\ full paper \cite{full}).} The neighbors of the deleted or inserted node are aware of the attack in the same time step and the self-healing algorithm responds by adding or dropping edges (i.e. connections) between nodes. The computation of the algorithm proceeds in synchronous rounds and we assume that the adversary does not perform any more changes until the algorithm has finished its response.
 As typical in self-healing (see e.g.
\cite{HayesFG-DCJournal-springerlink,PanduranganPODC11-abbrv,Amitabh-2010-PhdThesis}), we assume that no
other  insertion/deletion takes place during the repair phase \footnote{One way to think about this assumption is that insertion/deletion steps happen somewhat
at a slower time scale compared to the time taken by the self-healing algorithm to repair; hence this motivates the need to design fast self-healing algorithms.}
(though our algorithm can be potentially extended to handle such a scenario). The goal is to minimize the number of distributed rounds taken by the self-healing
algorithm to heal the network.

\noindent {\bf Our Contributions:} 
In this paper, we present \dex, in our knowledge the first  {\em 
distributed} algorithm to efficiently  construct and dynamically maintain
a constant degree expander network
(under both insertions and deletions) 
under an all-powerful adaptive adversary. Unlike previous constructions
(e.g.,\cite{PRU01,LS03-abbrv, mihail-p2p2006, AspnesW09,JacobSS-Skip09-abbrv}),
the expansion properties always hold, i.e.,  the algorithm guarantees that the
dynamic network {\em always} has a constant spectral gap (for some fixed
absolute constant) despite continuous network changes, and has constant degree,
and  hence is a (sparse) expander\onlyShort{\footnote{The full paper \cite{full} contains the formal definition of an expander and related concepts such as expansion, spectral gap, and Cheeger inequality.}}. The maintenance overhead
of \dex\ is very low. It uses only local information and 
small-sized messages, and hence is scalable. 
%
%
%
The following theorem  states our main result: 

\newcommand{\mainThm}{%
  Consider an adaptive adversary that observes the entire state of the
  network including all past random choices and inserts or removes
  a single node in every step.
  Algorithm \dex\ maintains a constant
  degree expander network that has a constant spectral gap. 
  The algorithm takes $O(\log n)$ rounds and messages   in the worst case (with
high probability\footnote{With high probability (w.h.p.) means with probability $\ge
  1-n^{-1}$.}) per insertion/deletion where $n$ is the current network size.
Furthermore,  \dex\ requires only a constant number of topology changes.}
\begin{theorem} \label{thm:main}
  \mainThm
\end{theorem}

Note that the above bounds hold w.h.p.\ for \emph{every} insertion/deletion (i.e., in a worst case sense)  and not just in
an amortized sense.
Our algorithm can be extended to handle multiple insertions/deletions per step
\onlyShort{ (cf.\ full paper \cite{full})}%
\onlyLong{ in the appendix (cf.\ Appendix~\ref{sec:multiplechanges})}. 
We also describe (cf.~Sec.~\ref{sec:dht}) how to implement a distributed hashtable (DHT) on top of our algorithm \dex, which provides insertion and lookup operations using $O(\log n)$ messages and rounds.

Our results  answer some open questions raised in prior work.
 In ~\cite{mihail-p2p2006}, the authors ask: Can
%
one can design a fully decentralized construction of dynamic expander topologies with {\em constant} overhead? 
 The expander maintenance  algorithms  of \cite{mihail-p2p2006} and
 \cite{LS03-abbrv} handle deletions much less effectively than additions;
 \cite{mihail-p2p2006} also  raises the question of handling deletions as effectively
 as insertions.  Our algorithm handles even {\em adversarial} deletions as
 effectively as insertions. 

\noindent {\bf Technical Contributions:} Our approach differs from previous approaches
to expander maintenance (e.g., \cite{LS03-abbrv, PRU01,mihail-p2p2006}). Our
approach \emph{simulates} a virtual network (cf.\ Sec.~\ref{subsec:
virtualexpander}) on the actual (real) network. 
At a high level, \dex\ works by stepping between instances of the guaranteed expander networks (of different sizes as required) in the virtual graph.
It maintains a {\em balanced  mapping} (cf.\
Def.~\ref{def:virtualmapping})  between the two networks  with the
guarantee  that the spectral
properties and degrees of both are similar. 
The virtual network is maintained as
a $p$-cycle expander (cf.\ Def.~\ref{def:prime}). Since the adversary is fully adaptive with
complete knowledge of topology and past random choices,  it is non-trivial to efficiently maintain {\em
both} constant degree and constant spectral gap of the virtual graph.   Our
maintenance algorithm \dex\ uses randomization to defeat the
adversary and exploits various key algorithmic properties of expanders,
in particular, Chernoff-like concentration bounds for random walks (\cite{G1998:Chernoff}),
fast (almost) uniform sampling, efficient permutation routing
(\cite{S1998:Universal-abbrv}), and the relationship between edge expansion and
spectral gap as stated by the Cheeger Inequality\onlyShort{ (cf.\ Theorem~2.6 in
  \cite{Wigderson-exsurvey})}\onlyLong{ (cf.\
  Theorem~\ref{thm:cheeger} in App.~\ref{app-sec:preliminaries})}. Moreover, we use certain structural properties of the $p$-cycle and staggering
of ``complex'' steps that require more involved recovery operations over multiple ``simple'' steps to achieve worst case $O(\log n)$ complexity bounds.
It is technically and conceptually much more convenient to work  on the
(regular) virtual network   and this can be a useful algorithmic paradigm
in handling other dynamic problems as well.


\noindent {\bf Related Work and Comparison:} 
\label{subsec: relatedworks}
Expanders are  a very important class of graphs that have applications in
various areas of computer science (e.g., see
~\cite{Wigderson-exsurvey} for a  survey) e.g. 
in distributed networks, expanders are used
for solving distributed agreement problems efficiently\cite{KingFOCS06-abbrv,
  APRU2012:Towards}.
In distributed dynamic networks (cf.\ \cite{APRU2012:Towards}) it is
particularly important that the expansion does not degrade over time.
There are many well known (centralized)
expander construction
techniques
see e.g., ~\cite{Wigderson-exsurvey}) .

As  stated earlier,   there are  a few other works addressing  the problem of distributed expander
construction; however all of these are randomized and the expansion
properties hold with probabilistic guarantees only.  Figure~\ref{tab: AlgoCompare}
compares our algorithm with some known distributed expander construction
algorithms. 
\cite{LS03-abbrv} give a construction where an expander is constructed by composing a small number of random Hamiltonian cycles. 
 The probabilistic guarantees provided degrade rapidly, especially under adversarial deletions.
\cite{mihail-p2p2006} builds on the algorithm of \cite{LS03-abbrv} and makes use of random walks to add new peers with only constant overhead. However, it is not a fully decentralized algorithm. Both these algorithms handle insertions much better than deletions. 
Spanders~\cite{DolevSpanders2010} is a self-stabilizing construction of an expander network that is a spanner of the graph. 
\cite{CooperFlipMarkovChainPODC09-abbrv} shows a way of
constructing random regular graphs (which are good expanders, w.h.p.) by performing a series of random `flip' operations on the graph's edges.
\cite{ReiterDistNetwork-SRDS2005-abbrv} maintains an almost $d$-regular graph, i.e. with degrees varying around $d$, using uniform sampling to select, for each node, a set of expander-neighbors. 
The protocol of \cite{PRU01} 
gives a distributed algorithm for maintaining a sparse random graph under a stochastic model of insertions and deletions.
\cite{MelamedK08} gives a dynamic overlay construction that is empirically shown to resemble a random k-regular graph and hence is a good expander. \cite{GurevichK10} gives a gossip-based membership protocol for maintaining an overlay in a dynamic network that under certain circumstances provides an expander.

In a model similar to ours, \cite{Kuhn2005-Repairing-abbrv} 
maintains a DHT (Distributed Hash Table) in the setting where an adaptive adversary can add/remove
$O(\log n)$ peers per step. Another paper which considers node
joins/leaves  is \cite{JacobSS-Skip09-abbrv} which constructs a SKIP+ graph within $O(\log^2 n)$ rounds starting from any graph whp. Then, they also show that after an insert/delete operation  the system recovers within $O(\log n)$ steps (like ours, which also needs $O(\log n)$ steps whp) and with $O(\log^4 n)$ messages (while ours takes $O(\log n)$ messages whp).
However, the SKIP+ graph has an additional advantage that it is {\em self-stabilizing},
i.e., can recover from any initial state (as long as it is weakly connected). 
\cite{JacobSS-Skip09-abbrv} assume (as do we) that the adversary rests while the network converges to a SKIP+ graph.
 It was shown in \cite{AspnesW09}
that skip graphs contain expanders as subgraphs w.h.p., which can be used
as a randomized expander construction.  Skip graphs (and its variant SKIP$+$~\cite{JacobSS-Skip09-abbrv}) 
are probabilistic structures (i.e., their expansion holds only with high
probability) and furthermore, they are not of constant degree, their
degree  grows logarithmic in the network size.  The work
of~\cite{DBLP:journals/talg/NaorW07} 
  has  guaranteed expansion (like ours). However, as pointed out in~\cite{AspnesW09}, its main drawback (unlike ours) is that their algorithm has a rather large overhead in maintaining the network.
 
 A variety of self-healing algorithms deal with maintaining topological invariants on arbitrary graphs~\cite{HayesFG-DCJournal-springerlink, PanduranganPODC11-abbrv,Amitabh-2010-PhdThesis,HayesPODC08-abbrv,SaiaTrehanIPDPS08-abbrv}.
The self-healing algorithm {\em
Xheal} of~\cite{PanduranganPODC11-abbrv} 
maintains  spectral properties of the network (while allowing only a small increase in
stretch and degree), but it  relied on a randomized expander construction
and hence the spectral properties  degraded rapidly.
Using our algorithm as a subroutine, {\em
Xheal}  can be efficiently implemented with guaranteed spectral properties.

%% file: Model.tex
\section{The Self-Healing Model}
\label{sec: model}
\onlyShort{\vspace{-0.3cm}}
 

The model we are using is similar to the models used in 
~\cite{HayesFG-DCJournal-springerlink,PanduranganPODC11-abbrv}.
We now describe the details.  Let $G = G_0$ be a small arbitrary
graph\footnotemark[\value{footnote}] 
where nodes represent processors in a distributed network and edges represent the links between them.
Each \emph{step} $t\ge 1$ is triggered by a deletion or insertion of a
single\onlyShort{\footnote{See the full paper \cite{full} for multiple insertions/deletions per step.}}\onlyLong{\footnote{See Appendix~\ref{sec:multiplechanges} for multiple 
    insertions/deletions per step.}} node from $G_{t-1}$ by the adversary, yielding an
\emph{intermediate network graph} $U_t$.
The neighbors of the (inserted or deleted) node in the network $U_{t}$
react to this change by adding or removing edges in $U_t$, yielding
$G_{t}$ --- this is called {\em recovery or repair}.
The distributed computation during recovery is structured into
synchronous rounds.
We assume that the adversary rests until the recovery is complete, and subsequently triggers the next step by inserting/deleting a node.
During recovery, nodes can communicate with their neighbors by sending messages of size $O(\log n)$, which are neither lost nor corrupted.
We assume that local computation (within a node) is free, which is a standard  assumption in distributed
computing (e.g.\ \cite{pelegDCbook}).
Our focus is only on the cost of communication (time and messages).

Initially, a newly inserted node $v$ only knows its unique id (chosen by the
adversary) and does not have any a priori knowledge of its
neighbors or the current network topology.
In particular, this means that a node $u$ can only add an edge to a node
$w$ if it knows the id of $w$.


In case of an insertion, we assume that the newly added node is initially
connected to a constant number of other nodes. This is merely a
simplification; nodes are not malicious but faithfully follow the
algorithm, thus we could explicitly require our algorithm to immediately
drop all but a constant number of edges.
The adversary is \emph{fully adaptive} and is aware of our algorithm, the
complete state of the current network including all past random choices.
As typically the  case (see e.g.\
\cite{HayesFG-DCJournal-springerlink,PanduranganPODC11-abbrv}), we assume that no
other  node is deleted or inserted until the current step has concluded
(though our algorithm can be modified to handle such a scenario).

%% file: Main.tex
\onlyShort{\vspace{-0.3cm}}
\section{Preliminaries and Overview of Algorithm \dex}
\label{sec: prelims}
\onlyShort{\vspace{-0.2cm}}

\begin{figure*}[t]



{\small
\begin{threeparttable}[b]
\begin{tabular}{|l|l|l|c|c|c|c|c|c|}
\hline
Algorithms & Expansion Guarantees& Adversary  & Max Degree & Recovery Time & Messages & Topology Changes \\
 \hline
 Law-Siu\cite{LS03-abbrv}\tnote{\$} & Prob$\ge 1-1/n_0$ & Oblivious &  $O(d)$ & $O(\log_d n)$    & $O(d\log_d n)$  & $O(d)$ \\
   \hline
    Skip Graphs~\cite{AspnesW09}\tnote{$\ddagger$}& w.h.p.\tnote{$\dagger$} & Adaptive  &  $O(\log n)$ &$O(\log^2 n)  $  & $ O(\log^2 n) $   &  $O(\log n)$  \\
 \hline
 Skip+~\cite{JacobSS-Skip09-abbrv}\tnote{$!$} & w.h.p.\tnote{$\dagger$} &  Adaptive & $ O(\log n) $ &$O(\log n) \tnote{$\dagger$} $  & $ O(\log^4 n) $   &  $O(\log^4 n) \tnote{$\dagger$}$  \\
 \hline
\dex\ (This paper) & Deterministic  &  Adaptive & $O(1)$ &  $O(\log n) \tnote{$\dagger$}$ & $ O(\log n )$\tnote{$\dagger$}  & $O(1)$  \\
  \hline
 \end{tabular}
 \begin{tablenotes} 
 \item[$\dagger$] With high probability.
   \item[\$] $n_0$ is the initial network size.
Parameter $d$ = \# of Hamiltonian cycles in 'healing' graph ($\mathbb{H}$). 
   \item[$\ddagger$] Costs given under certain assumptions about key length.
  \item[$!$] Skip+ is a self-stabilizing structure but  costs here are for single join/leave operations once a valid skip+ graph is achieved.
\end{tablenotes} 
\caption{\small Comparison of distributed expander constructions.  }
\end{threeparttable}

\label{tab: AlgoCompare}
} 

\end{figure*}

It is instructive to first consider the
following natural (but inefficient) algorithms:

\noindent\textbf{Flooding:}
First, we consider a naive flooding-based algorithm that also achieves
guaranteed expansion and node degree bounds, albeit at a much larger cost:
Whenever a node is inserted (or
deleted), a neighboring node floods a notification throughout the entire network
and every node, having complete knowledge of the current network graph, 
locally recomputes the new expander topology.  While this achieves a logarithmic
runtime bound, it comes at the cost of using $\Theta(n)$ messages in
\emph{every} step and, in addition, might also result in $O(n)$ topology
changes, whereas our algorithms requires only polylogarithmic number of messages
and constant topology changes on average.

\noindent\textbf{Maintaining Global Knowledge:} As a second example of a 
straightforward but inefficient solution,
consider the algorithm that maintains a global knowledge at some node $p$, 
which keeps track of the entire network topology.
Thus, every time some node $u$ is inserted or deleted, the neighbors of
$u$ inform $p$ of this change, and $p$ then proceeds to
update the current graph using its global knowledge.
However, when $p$ itself is deleted, we would need to transfer all of its
knowledge to a neighboring node $q$, which then takes over $p$'s role.
This, however, requires at least $\Omega(n)$ rounds, since the 
\emph{entire} knowledge of the network topology needs to be transmitted to 
$q$.

\noindent\textbf{Our Approach --- Algorithm \dex:} 
As  mentioned in Sec.\ \ref{sec: model}, the actual (real) network is
represented by a graph where nodes correspond to processors and edges to
connections.
Our algorithm maintains a second graph, which we call the {\em virtual
  graph} where the vertices do not directly correspond to the real network but each (virtual)
vertex in this graph is simulated by a (real) node \footnote{Henceforth, we reserve the term ``vertex'' for vertices in a virtual graph and
(real) ``node'' for vertices in the real network.} in the network.
The topology of the virtual graph  determines the connections in the actual
network.
For example, suppose that node $u$ simulates vertex $z_1$ and node $v$
simulates vertex $z_2$.
If there is an edge $(z_1,z_2)$ according to the virtual graph, then our
algorithm maintains an edge between $u$ and $v$ in the actual network.
In other words, a real node may be simulating multiple virtual vertices
and maintaining their edges according to the virtual graph.


Figure~\ref{fig:graphs} on page~\pageref{fig:graphs} shows a real network (on the right) whose nodes
(shaded rectangles) simulate the virtual vertices of the virtual
graph (on the left).
\onlyShort{
\begin{figure}[h]
\centering
  \includegraphics[scale=0.60]{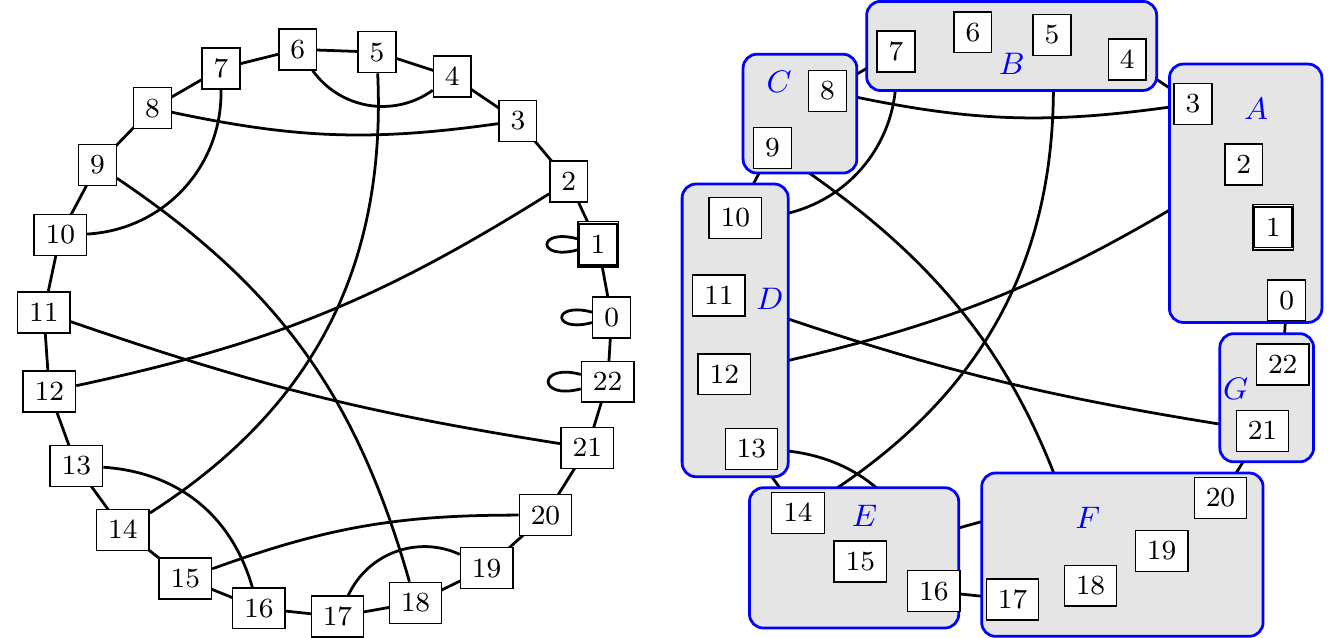}
\caption{ A $4$-balanced virtual mapping of a $p$-cycle expander to the
network graph. On the left is a (virtual) $3$-regular $23$-cycle expander
on $\mathbb{Z}_{23}$; on the right is the network $G_t$ with (real) nodes $\{A,\dots,G\}$.}
\label{fig:graphs}
\end{figure}
}
In our algorithm, we maintain this
virtual graph and show that preserving certain desired properties (in particular, constant expansion and degree) in the
virtual graph leads to these properties being preserved in the real network. 
%
Our algorithm achieves this by
maintaining a ``balanced load mapping'' (cf.\ Def.\ \ref{def:cbalanced}) between the virtual vertices and the
real nodes as the network size changes at the will of the adversary. The balanced load mapping
keeps the number of virtual nodes simulated by any real node to be a constant --- this is crucial
in maintaining the constant degree bound. 
We next formalize the notions of virtual graphs and balanced mappings.
\onlyShort{\vspace{-0.1in}}
\subsection{Virtual Graphs and Balanced Mappings} \label{subsec: virtualexpander}
\onlyLong{
\begin{figure}[h]
\centering
  \includegraphics[scale=0.6]{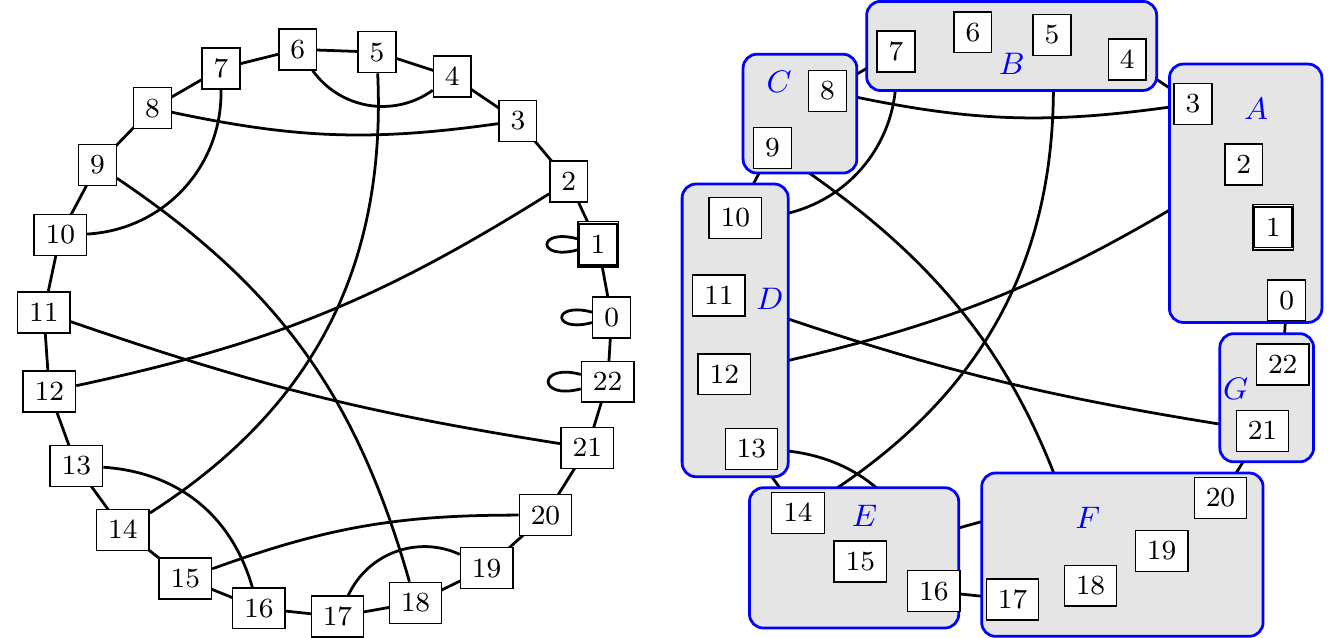} 
\caption{ A $4$-balanced virtual mapping of a $p$-cycle expander to the
network graph. On the left is a (virtual) $3$-regular $23$-cycle expander
on $\mathbb{Z}_{23}$; on the right is the network $G_t$ with (real) nodes $\{A,\dots,G\}$.}
\label{fig:graphs}
\end{figure}
}

Consider some graph $G$ and let $\lambda_G$ denote the \emph{second largest
  eigenvalue} of the adjacency matrix of $G$.
The \emph{contraction} of vertices $z_1$ and $z_2$ produces a graph $H$
where $z_1$ are $z_2$ merged into a single vertex $z$ that is adjacent to
all vertices to which $z_1$ or $z_2$ were adjacent in $G$.
We extensively make use of the fact that this operation leaves the
\emph{spectral gap} $1-\lambda_G$ intact, cf.\onlyShort{ Lemma~1.15 in \cite{chungbook}}\onlyLong{ Lem.~\ref{lem:contraction} in App.~\ref{app-sec:preliminaries}}.

As mentioned earlier, our virtual graph consists of virtual vertices
simulated by real nodes. 
Intuitively speaking, we can think of a real node simulating $z_1$ and
$z_2$ as a vertex contraction of $z_1$ and $z_2$.
The above stated contraction property motivates us to use an expander family\onlyLong{ (cf.\
  Def.\ \ref{def:expanders} in App.~\ref{app-sec:preliminaries})} as virtual
graphs.
We now define the $p$-cycle expander family, which we use as virtual
graphs in this paper.
Essentially, we can think of a $p$-cycle as a numbered cycle with some
chord-edges between numbers that are multiplicative inverses of each
other.
It was shown in \cite{L1994:Discrete} that this yields an infinite family
of $3$-regular expander graphs with a constant eigenvalue gap.
Figure~\ref{fig:graphs} shows a $23$-cycle.
\begin{definition}[{\boldmath \bf $p$-cycle}, cf.\ \cite{Wigderson-exsurvey}] \label{def:prime}
For any prime number $p$, we define the following graph $\Z(p)$. 
The vertex set of $\Z(p)$ is the set $\mathbb{Z}_p=\{0,\dots,p-1\}$ and there is an edge between vertices $x$ and $y$ if and only if one of the following 
conditions hold:
(1) $y = (x + 1) \mod p$,
(2) $y = (x - 1) \mod p$, or
(3) if $x,y>0$ and $y = x^{-1}$.
Moreover, vertex $0$ has a self-loop. 
\end{definition}

At any point in time $t$, our algorithm maintains a mapping from the virtual
vertices of a $p$-cycle to the actual network nodes.
We use the notation $\Z_t(p)$ when $\Z(p)$ is the $p$-cycle that we are using
for our mapping in step $t$.
(We omit $p$ and simply write $\Z_t$ if $p$ is irrelevant or clear from the
context.)
At any time $t$, each real node simulates at least one virtual vertex (i.e.\ a
vertex in the $p$-cycle) and all its incident edges as required by Def.\
\ref{def:prime}, i.e., the real network
can be considered a contraction of the virtual graph; see
Figure~\ref{fig:graphs} on page~\pageref{fig:graphs} for an example.
Formally, this defines a function that we call a virtual mapping:

\begin{definition}[\bf Virtual mapping] \label{def:virtualmapping}
For step $t\ge 1$, consider a surjective map $\Phi_{t} : V(\Z_t)
\rightarrow V(G_t)$ that maps every {virtual vertex} of the virtual graph
$\Z_t$ to some {(real) node} of the network graph $G_t$.
Suppose that there is an edge $(\Phi_t(z_1), \Phi_t(z_2)) \in E(G_t)$ if
and only if there is an edge $(z_1,z_2) \in E(Z_t)$, for all nodes $z_1, z_2 \in V(Z_t)$.
Then we call $\Phi_{t}$ a \emph{virtual mapping}.
Moreover, we say that node $u \in V(G_t)$ is a real node that
\emph{simulates} virtual vertices $z_1,\cdots,z_k$, if
$u=\Phi_t(z_1)=\cdots=\Phi_t(z_k)$.
\end{definition}
In the standard metric spaces on $\Z_t$ and $G_t$ induced by the
shortest-path metric $\Phi$ is a surjective {metric map} since distances
do not increase:
\begin{fact} \label{fact:distances}
  Let $\dist_H(u,v)$ denote the length of the shortest path between
  $u$ and $v$ in graph $H$.
  Any virtual mapping $\Phi_t$ guarantees that $\dist_{Z_t}(z_1,z_2) \ge
  \dist_{G_t}(\Phi(z_1),\Phi(z_2))$, for all $z_1, z_2 \in \Z_t$.
\end{fact}
We simply write $\Phi$ instead of $\Phi_t$ when $t$ is irrelevant.

We consider the vertices of $\Z_t$ to be partitioned into disjoint sets of
vertices that we call \emph{clouds} and denote the cloud to which a vertex
$z$ belongs as $\cloud{z}$.
Whereas initially we can think of a cloud as the set of virtual vertices
simulated at some node in $G_t$, this is not true in general due to load
balancing issues, as we discuss in Section~\ref{sec:detexpalgo}.
We are only interested in virtual mappings where the
\emph{maximum cloud size} is bounded by some universal constant $\zeta$,
which is crucial for maintaining a constant node degree.
For our $p$-cycle construction, it holds that $\zeta \le 8$.

We now formalize the intuition that the expansion of the
virtual $p$-cycle carries over to the network, i.e., the second largest
eigenvalue $\lambda_{G_t}$ of the real network is bounded by
$\lambda_{Z_t}$ of the virtual graph.
Observe that we can obtain $G_t$ from $\Z_t$ by
  contracting vertices. That is, we contract vertices $z_1$ and $z_2$
  if $\Phi(z_1)=\Phi(z_2)$. 
\onlyShort{According to Lemma~1.15 in \cite{chungbook}}\onlyLong{According to Lemma~\ref{lem:contraction},} these operations do not
  increase $\lambda_{G_t}$ and thus we have shown the following:

\newcommand{\lemFanChung}{
  Let $\Phi_t:\Z_t\ra G_t$ be a virtual mapping. Then it holds that
    $\lambda_{G_t} \le \lambda_{\Z_t}$.
}
\begin{lemma}
  \label{lem:fanChung}
  \lemFanChung
\end{lemma}

Next we formalize the notion that our real nodes simulate at most a constant number of nodes.
Let $\Sim_{t}(u)=\Phi_t^{-1}(u)$ and define the \emph{load of a node $u$
  in graph $G_t$} as the number of vertices simulated at $u$, i.e., 
$\Load_{t}(u)=|\Sim_{t}(u)|$.
Note that due to locality, node $u$ does not necessarily know the mapping
of other nodes.
\begin{definition}[\bf Balanced mapping]\label{def:cbalanced}
Consider a step $t$.
If there exists a constant $C$ s.t.\
$\forall u \in G_t\colon \Load_{t}(u) \le C, $
then we say that $\Phi_{t}$ is a \emph{$C$-balanced virtual mapping} and 
say that \emph{$G_t$ is $C$-balanced}.
\end{definition}
Figure~\ref{fig:graphs} on page~\pageref{fig:graphs} shows a balanced virtual mapping.
At any step $t$, the degree of a node $u\in G_t$ is exactly
$3.\Load_t(u)$ since we are using the $3$-regular $p$-cycle as a virtual graph. 
Thus our algorithm strives to maintain a constant bound on $\Load_t(u)$, for all $t$.
Given a virtual mapping $\Phi_t$, we define the
\onlyLong{(not necessarily disjoint) }sets
\begin{align}
  \Low_{t} &= \{ u \in G_t\colon \Load_{t}(u) \le 2\zeta \}; \label{eq:low}\\
  \High_{t} &= \{ u \in G_t\colon \Load_{t}(u) \ge 2 \}. \label{eq:spare}
\end{align}
%
Intuitively speaking, $\Low_{t}$ contains nodes that do not
simulate too many virtual vertices, i.e., have relatively low degree, whereas
$\High_t$ is the set of nodes that simulate at least $2$ vertices each. When the
adversary deletes some node $u$, we need to find a node in $\Low_t$ that takes
over the load of $u$. Upon a node $v$ being inserted, on the other hand, we
need to find a node in $\High_t$ that can spare a virtual vertex for $v$,
while maintaining the surjective property of the virtual mapping.

\section{Expander Maintenance Algorithm}
\label{sec:detexpalgo}
\onlyShort{\vspace{-0.1in}}
We  describe our maintenance algorithm \dex\onlyShort{ (the complete pseudo-code is in the full paper \cite{full})} and prove\onlyShort{ (omitted proofs are included 
  in the full paper \cite{full})} the
performance claims of Theorem~\ref{thm:main}.
We start with a small initial network $G_0$ of some appropriate constant
and assume there is a virtual mapping from a $p$-cycle
$\Z_0(p_0)$ where $p_0$ is the smallest prime number in the range
$(4n_0,8n_0)$.
The existence of $p_0$ is guaranteed by Bertrand's
postulate~\cite{Bertrand1850}.
(Since $G_0$ is of constant size, nodes can compute the current
network size $n_0$ and $\Z_0(p_0)$ in a constant number of rounds in a
centralized manner.)
Starting out from this initial expander, we seek to guarantee expansion
ad infinitum, for any number of adversarial insertions and deletions.

We always maintain the invariant that each real node {simulates} at
least one (i.e.\ the virtual mapping is surjective) and at most a constant number of virtual $p$-cycle vertices.
The adversary can either insert or delete a node in every step.
In either case, our algorithm reacts by doing an appropriate
redistribution of the virtual vertices to the real nodes with the goal of
maintaining a $C$-balanced mapping (cf.\ Definition~\ref{def:cbalanced}). 

Depending on the operations employed by the algorithm, we classify the response
of the algorithm for a given step $t$ as being either a \emph{type-1 recovery} or a
\emph{type-2 recovery} and call $t$ a \emph{type-1 recovery step} (resp.\
type-2 recovery step).
Type-1 recovery is very efficient, as (w.h.p.) it suffices to execute a single random walk of $O(\log n)$ length.

It is somewhat more complicated to show a worst case $O(\log n)$ performance
for type-2 recovery: Here, the current virtual graph is either \emph{inflated} or
\emph{deflated} to ensure a $C$-balanced mapping (i.e.\ bounded degrees).
For the sake of exposition, we  first present a simpler way to handle inflation and deflation, which yields amortized complexity bounds.
We then describe a more complicated algorithm for type-2 recovery that yields the claimed {\em worst case} complexity bounds of $O(\log n)$ rounds and messages, and $O(1)$ topology changes per step with high probability.
The first (simplified) approach (cf.\ Sec.~\ref{sec:infldefl}) replaces the
entire virtual graph by a new virtual graph of appropriate size in a single step.
This requires $O(n)$ topology changes and $O(n\log^2 n)$ message complexity,
because all nodes complete the inflation/deflation in one step.
Since there are at least $\Omega(n)$ steps with type-1 recovery between any
two steps where inflation or deflation is necessary, we can nevertheless amortize their cost and get the amortized performance bounds of $O(\log n)$ rounds and $O(\log^2 n)$ messages (cf.\ Cor.~1\onlyShort{ in the full paper \cite{full})}\onlyLong{)}.
We then present an improved (but significantly more complex) way of handling inflation (resp.\ deflation), by \emph{staggering} these
inflation/deflation operations across the recovery of the next $\Theta(n)$ following steps while retaining constant expansion and node degrees.
This yields a {\em $O(\log n)$ worst case bounds for both messages and rounds} for \emph{all} steps as claimed by Theorem~\ref{thm:main}.
In terms of expansion, the (amortized) inflation/deflation approach yields a
spectral gap no smaller than of the $p$-cycle, the improved worst case bounds of
the 2nd approach come at the price of a slightly reduced, but still constant,
spectral gap.
Algorithm~\ref{algo:highlevel} presents a high-level pseudo code description of
our approach.

\onlyShort{
\begin{algorithm}[h!]
  \small
\begin{algorithmic}
\item[] 
\item[] \textbf{Case 1: Adversary inserts a node $u$}:
\STATE Try to find a spare vertex for $u$ via a random walk
(\textbf{type-1 recovery}). 
\IF{type-1 recovery fails} 
  \IF{most nodes simulate only $1$ vertex} 
    \STATE Perform \textbf{type-2 recovery} by inflating.
  \ELSE
    \STATE Retry type-1 recovery until it succeeds.
  \ENDIF
\ENDIF
\item[] 
\item[] \textbf{Case 2: Adversary deletes a node $u$}:
\STATE Try distributing vertices that were simulated at $u$ via random walks
(\textbf{type-1
recovery}). 
\IF{type-1 recovery fails} 
  \IF{most nodes simulate many vertices} 
    \STATE Perform \textbf{type-2 recovery} by deflating.
  \ELSE
    \STATE Retry type-1 recovery until it succeeds.
  \ENDIF
\ENDIF
\end{algorithmic}
\caption{High-level overview of our algorithm \dex}
\label{algo:highlevel}
\end{algorithm}
}

\onlyShort{\vspace{-0.1in}}
\subsection{Type-1 Recovery} \label{sec:simple}
\onlyShort{\vspace{-0.1in}}

When a node $u$ is inserted, a neighboring node $v$ initiates a random walk of
length at most $\Theta(\log n)$ to find a ``spare'' virtual vertex, i.e., a
virtual vertex $z$ that is simulated by a node $w \in \High_{G_{t-1}}$\onlyLong{ (see
  Algorithm~\ref{app-algo:insertion} for the detailed pseudo code)}. 
Assigning this virtual vertex $z$ to the new node $u$, ensures a
surjective mapping of virtual vertices to real nodes at the end of the
step.

When a node $u$ is deleted, on the other hand, the notified neighboring node $v$
also initiates random walks, except this time with the aim of redistributing the
deleted node $u$'s virtual vertices to the remaining real nodes in the system\onlyLong{(cf.\ Algorithm~\ref{app-algo:deletion})}.
We assume that every node $v$ has knowledge of $\Load_{G_{t-1}}(w)$, for each of
its neighbors $u$.
(This can be implemented with constant overhead, by simply updating neighboring
nodes when the respective $\Load_{G_{t-1}}$ changes.)
Since the deleted node $u$ might have simulated multiple vertices,  node
$v$ initiates a random walk for each $z \in \Load_{G_{t-1}}(u)$, to find
a node $w \in \Low_{G_{t-1}}$ to take over virtual vertex $z$.
In a nutshell, type-1 recovery consists of (re)balancing the load of virtual
vertices to real nodes by performing random walks. 
Rebalancing the load of a deleted node succeeds with high probability, as long
as $\theta n$ nodes are in $\Low_{G_{t-1}}$,  where the \emph{rebuilding parameter} $\theta$ is a fixed constant.
For our analysis, we require that
\onlyLong{\begin{equation} \label{eq:theta}}
\onlyShort{$}
  \theta \le 1/(68\zeta +1),
\onlyShort{$}
\onlyLong{\end{equation}}
where $\zeta \le 8$ is the maximum (constant) cloud size given by the $p$-cycle
construction.
Analogously, for insertion steps, finding a spare vertex will succeed w.h.p.\ if
$\High_{G_{t-1}}$ has size $\ge \theta n$.
If the size is below $\theta n$, we handle the insertion (resp.\ deletion) by
performing an inflation (resp.\ deflation) as explained below.
Thus we  formally define a step $t$ to be a \emph{type-1 step}, if either 
\onlyShort{(1) a node is inserted in $t$ and $|\High_{G_{t-1}}|\ge \theta n$ or 
(2) a node is deleted in $t$ and $|\Low_{G_{t-1}}|\ge \theta n$.}%
\onlyLong{%
\begin{compactenum}
\item[(1)] a node is inserted in $t$ and $|\High_{G_{t-1}}|\ge \theta n$ or 
\item[(2)] a node is deleted in $t$ and $|\Low_{G_{t-1}}|\ge \theta n$.
\end{compactenum}
}

If a random walk fails to find an appropriate node, we do not directly
start an inflation resp.\ deflation, but first {deterministically} count the
network size and sizes of $\High_{G_{t-1}}$ and $\Low_{G_{t-1}}$ by simple
aggregate flooding (cf.\ Procedures $\computeLow$ and $\computeHigh$).
We repeat the random walks, if it turns out that the respective set indeed
comprises $\ge \theta n$ nodes.
As we will see below, this allows us to deterministically guarantee constant node degrees.
\onlyLong{%
The following lemma shows an $O(\log n)$ bound for messages and rounds used by
random walks in type-1 recovery:}%
\onlyShort{We will only consider the case where a node is deleted and defer the (analogously handled) insertions to the full paper \cite{full}.}
\newcommand{\lemGillman}{
  Consider a step $t$ and suppose that
  $\Phi_{t-1}$ is a $4\zeta$-balanced virtual map.
There exists a constant $\ell$ such that the following hold w.h.p:
\onlyLong{%
\begin{compactitem}
\item[(a)] If $|\High_{G_{t-1}}|\ge \theta n$ and
  a new node $u$ is attached to some node $v$,
  then the random walk initiated by $v$ reaches a node in
  $\High_{G_{t-1}}$ in $\ell\log n$ rounds.
\item[(b)] If $|\Low_{G_{t-1}}|\ge \theta n$ and
  some node $u$ is deleted,
  then, for each of the (at most $4\zeta\in O(1)$) vertices simulated at 
  $u$, the initiated random walk reaches a node in $\High_{G_{t-1}}$ in 
  $\ell\log n$ rounds.
\end{compactitem}
}
\onlyShort{%
If $|\Low_{G_{t-1}}|\ge \theta n$ and
  some node $u$ is deleted,
  then, for each of the (at most $4\zeta\in O(1)$) vertices simulated at 
  $u$, the initiated random walk reaches a node in $\High_{G_{t-1}}$ in 
  $\ell\log n$ rounds.}
That is, w.h.p.\ type-1 recovery succeeds in $O(\log n)$ messages and rounds, and 
a constant number of edges are changed.
}
\begin{lemma} \label{lem:gillman}
  \lemGillman
\end{lemma}
\begin{proof}
  \onlyLong{We will first consider the case where a node is deleted (Case (b)).}
  The main idea of the proof is to instantiate a concentration bound for
  random walks on expander graphs \cite{G1998:Chernoff}. 
  By assumption, the mapping of virtual vertices to real nodes is
  $4\zeta$-balanced before the deletion occurs.
  Thus we only need to redistribute a constant number of
  virtual vertices when a node is deleted.

  We now present the detailed argument.
  By assumption we have that
  $|\Low| = a n \ge \theta n$, for a constant $0< a< 1$.
  We start a random walk of length $\ell\log
  n$ for some appropriately chosen constant $\ell$ (determined below).  
  We need to show that (w.h.p.) the walk hits a node in $\Low$.
According to the description of type-1 recovery for handling deletions, we perform
the random walk on the graph $G'_t$, which modifies $G_{t-1}\setminus\{u\}$, by
transferring all virtual vertices (and edges) of the deleted node $u$ to the
neighbor $v$.
Thus, for the second largest eigenvalue $\lambda=\lambda_{G'_{t}}$, we know
by Lemma~\ref{lem:fanChung} that $\lambda \le \lambda_{G_{t-1}}$.
Consider the normalized $n\times n$ adjacency matrix $M$ of $G'_{t}$.
It is well known (e.g., Theorem~7.13 in \cite{MU2005:Probability}) that a
vector $\pi$ corresponding to the stationary distribution of a random walk
on $G_{t-1}$ has entries $\pi(x) = \frac{d_x}{2|E(G'_t)|}$ where $d_x$ is
the degree of node $x$.
By assumption, the network $G_{t-1}$ is the image of a $4\zeta$-balanced 
virtual map. This means that the maximum degree $\Delta$ of any node in the 
network is $\Delta \le 12 \zeta$, and since the $p$-cycle is a $3$-regular
expander, every node has degree at least $3$.
If the adversary deletes some node in step $t$, the maximum degree of 
one of its neighbors can increase by at most $\Delta$.
Therefore, the maximum degree in $U_t$ and thus $G'_t$ is bounded by 
$2\Delta$, which gives us the bound
\begin{equation} \label{eq:pi1}
  \pi(x) \ge {3}/({2\Delta n}),
\end{equation}
for any node $x \in G'_t$.
  Let $\rho$ be the actual number of nodes in $\Low$ that the random walk of
  length $\ell\log n$ hits.
  We define $\mathbf{q}$ to be an 
  $n$-dimensional vector that is $0$ everywhere except at the index of $u$ in 
  $M$ where it is $1$.  
  Let $\mathcal{E}$ be the event that 
  $\ell\log n \cdot \pi(\Low) - \rho \ge \gamma,$ for
  a fixed $\gamma\ge 0$.
  That is, $\mathcal{E}$ occurs if the number of nodes in $\Low$ visited by
  the random walk is far away  ($\ge \gamma$) from its expectation.

  In the remainder of the proof, we show that $\mathcal{E}$ occurs
  with very small probability. 
  Applying the concentration bound of
  \cite{G1998:Chernoff} yields that
  \onlyShort{$}
  \onlyLong{  \begin{equation}
      \label{eq:gillman1}}
  \Pr\left[\mathcal{E}\right] \le
\left(1+\frac{\gamma(1-\lambda)}{10\ell\log n}\right)\cdot\left|\!\left|\frac{\mathbf{q}}{\sqrt{\pi}}\right|\!\right|_2 
\cdot e^{\frac{-\gamma^2(1-\lambda)}{20\ell\log n}},
\onlyLong{\end{equation}}
\onlyShort{$}
  where $\mathbf{q}/\sqrt{\pi}$ is a vector with entries
  $(\mathbf{q}/\sqrt{\pi})(x)=\mathbf{q}(x)/\sqrt{\pi(x)}$, for $1\le x\le n$.
  By \eqref{eq:pi1}, we know that $\pi(\Low) \ge  3 a / 2\Delta$.
  To guarantee that we find a node in $\Low$ w.h.p.\ even when $\pi(\Low)$
  is small, we must set $\gamma = \frac{ 3 a \ell}{2\Delta} \log n $.
  Moreover, 
  \eqref{eq:pi1} also gives us the bound $|\!|{\mathbf{q}/\sqrt{\pi}}|\!|_2 
  \le \sqrt{2\Delta/3}\sqrt{n}.$
  We define $C=\left(1 + \frac{ 3 a}{20\Delta}\right)\sqrt{2\Delta/3}.$
  Plugging these bounds into\onlyLong{ \eqref{eq:gillman1}}\onlyShort{ the bounds on $\Pr[\mathcal{E}]$}, shows that
$
    \Pr\left[\mathcal{E}\right] 
    \le
    C\sqrt{n} e^{\left(-\frac{(3 a  \ell / 2\Delta)^2 
          (1-\lambda)\log n}{20\ell}\right)} 
    = C n^{\left(\frac{1}{2}-\frac{9 a^2  \ell (1-\lambda) }{ 80\Delta^2}\right)}. 
$
  To ensure that event $\mathcal{E}$ happens with small probability, it is
  sufficient if the exponent of $n$ is smaller than $-C$,
  which is true for sufficiently large $\ell$.
  Since $\theta$, $\Delta$, and the spectral gap $1-\lambda$ are all
  $O(1)$, it follows that $\ell$ is a constant too and thus the running
  time of one random walk is $O(\log n)$ with high probability.
  Recall that node $v$ needs to perform a random walk for each of the
  virtual vertices that were previously simulated by the deleted node $u$;
  there are at most $4\zeta \in O(1)$ such vertices, since we
  assumed that $\Phi_{t-1}$ is $4\zeta$ balanced.
  Therefore, all random walks take $O(\log n)$ rounds in total (w.h.p.).

  \onlyLong{
  Now consider Case (a), i.e., the adversary inserted a new node $u$ and
  attached it to some existing node $v$.
  By assumption, $|\High| = a n \ge \theta n$, and the random walk is
  executed on the graph $G_{t-1}$ (excluding newly inserted node $u$).
  Thus \eqref{eq:pi1} and the remaining analysis hold analogously to Case (b),
  which shows that the walk reaches a node in $\High$ in $O(\log n)$ rounds (w.h.p.).
}

  Note that we only transfer a constant number of virtual vertices to a new nodes in type-1 recovery steps, i.e., the number of topology changes is constant.
\end{proof}
\onlyLong{
The following lemma summarizes the properties that hold after performing a
type-1 recovery: 
}
\newcommand{\lemAtMostSimple}{
  If type-1 recovery is performed in $t$ and $G_{t-1}$ is $4\zeta$-balanced,
  it holds that 
\onlyShort{\vspace{-0.05in}}
  \begin{compactitem}
    \item[(a)] $G_t$ is $4\zeta$-balanced,
    \item[(b)] step $t$ takes $O(\log n)$ (w.h.p.),
      rounds, 
    \item[(c)] nodes send $O(\log n)$ messages in step $t$ (w.h.p.),
      and
    \item[(d)] the number of topology changes in $t$ is constant.
  \end{compactitem}
\onlyShort{\vspace{-0.1in}}
}
\begin{lemma}[Worst Case Bounds Type-1 Rec.] \label{lem:atMostSimple}
  \lemAtMostSimple
\end{lemma}
\onlyLong{
\begin{proof}
  For (a), we first argue that the mapping $\Phi_t$ is surjective: This follows
  readily from the above description of type-1 recovery (see
  $\insertion(u,\theta)$ and $\deletion(u,\theta)$ for the full pseudo code):
  In the case of a newly inserted node, the algorithm repeatedly performs a
  random walk until it finds a node in $\High$ since $|\High|\ge \theta n$.
  If some node $u$ is deleted, then a neighbor initiates random walks to find a
  new host for each of $u$'s virtual vertices, until it succeeds.
  Thus, at the end of step $t$, every node simulates at least $1$ virtual
  vertex.
  To see that no node simulates more than $4\zeta$ vertices, observe that the
  load of a node can only increase due to a deletion.
  As we argued above, however, the neighbor $v$ that temporarily took over the
  virtual vertices of the deleted node $u$, will attempt to spread these vertices to nodes
  that are in $\Low$ and is guaranteed to eventually find such nodes by
  repeatedely performing random walks.

Properties (b), (c), and (d) follow from Lemma~\ref{lem:gillman}.
\end{proof}
}

\onlyShort{
\begin{algorithm}[h!]
  \small
\begin{algorithmic}[1]
\item[]{\textbf{Assumption:} the adversary attaches inserted node $u$ to
    arbitrary node $v$.}
\item[] 
\item[] \COMMENT{Try to perform a \textbf{type-1 recovery}:}
\STATE Node $v$ initiates a random walk of length $\ell\log n$ by
generating a token $\tau$ and sending it to a neighbor $u'$ chosen
uniformly at random, but excluding $u$.
Node $u'$ in turn forwards $\tau$ by chosing a neighbor at random and
so forth.
Note that the newly inserted node $u$ is excluded from being reached by
the random walk.
The walk terminates upon reaching a node $w \in \High$ (cf.\
Equation~\eqref{eq:spare}).
\label{line:sampleHighDegree}

\IF{found node $w \in \High$}
  \STATE Transfer a virtual vertex and all its edges (according to
  the virtual graph) from $w$ to $u$. Remove edge between $u$ and $v$ unless
  required by $\Z_t$.
\item[]
\ELSE[the walk did not hit a node in $\High$; perform \textbf{type-2 recovery} if necessary:]
  \STATE Determine current network size $n$ and $|\High|$ via
  $\computeHigh$ (cf.\ Algorithm~\ref{algo:computeHighLow}).
  \IF[Perform type-2 recovery:]{$|\High| < \theta n$}
  \onlyLong{\STATE Invoke $\inflate$ (cf.\ Algorithm~\ref{algo:inflate}).}
  \onlyShort{\STATE Invoke type-2 recovery.}
  \ELSE[Sufficiently many nodes with spare virtual vertices are present but
  the walk did not find them. Happens with probability $\le 1/n$.]
    \STATE Repeat from Line~\ref{line:sampleHighDegree}.
  \ENDIF 
\ENDIF
\end{algorithmic}
\caption{\texttt{insertion($u, \theta$)}}
\label{app-algo:insertion}
\end{algorithm}

\begin{algorithm}[h!]
  \small
\begin{algorithmic}[1]
\item[] \textbf{Assumption:} adversary deletes an arbitrary node $u$ which
  simulated $k$ virtual vertices. (We prove that $k \in O(1)$).
\STATE A (former) neighbor $v$ of node $u$ attaches all edges of $u$ to itself.

\item[] 
\item[] \COMMENT{Try to perform a \textbf{type-1 recovery}:}
\FOR{each of the $k$ vertices} 
  \STATE Node $v$ initiates a random walk of length $\ell\log n$ by
  generating a token $\tau$ and sending it to a uniformly at random chosen
  neighbor $u'$.
  Node $u'$ in turn forwards $\tau$ by chosing a neighbor at random and
  so forth.
  The walk terminates upon reaching a node $w \in \Low$ (cf.\
  Equation~\eqref{eq:low}).
\ENDFOR
\label{line:ifSampleLowDegree}
\IF{all random walks found nodes $w_1,\dots,w_k \in \Low$:}
  \STATE Distribute the virtual vertices of $u$ and their respective edges (according to
  the virtual graph) from $v$ to $w_1,\dots,w_k$. 
\item[]
\ELSE[Some of the random walks did not find a node in $\Low$; perform
\textbf{type-2 recovery} if necessary:]
  \STATE Determine network size $n$ and $|\Low|$ via $\computeLow$ (cf.\
  Algorithm~\ref{algo:computeHighLow}).
  \IF[Perform type-2 recovery:]{$|\Low| < \theta n$}
    \onlyLong{\STATE Invoke $\deflate$ (cf.\ Algorithm~\ref{algo:deflate}).}
    \onlyShort{\STATE Invoke type-2 recovery.}
  \ELSE[Sufficiently many nodes with low load are present but the walk(s) did not find
  them. This happens with probability $\le 1/n$:]
    \STATE  Repeat from Line~\ref{line:ifSampleLowDegree}.
  \ENDIF 
\ENDIF
\end{algorithmic}
\caption{Procedure \texttt{deletion($u,\theta$)}}
\label{app-algo:deletion}
\end{algorithm}

\begin{algorithm}[h!]
  \small
\begin{algorithmic}[1]
\item[] \textbf{Given:} $\textsc{diam}$ is the diameter of  $\Z_t$ (i.e. $\textsc{diam} \in O(\log n)$).
\STATE Node $u$ broadcasts an aggregation request to all its neighbors. In
addition to the network size, this request indicates whether to compute
$|\Low|$ or $|\High|$. That is, the request of $u$ traverses the network in
a BFS-like manner and then returns the aggregated values to $u$.
\STATE If a node $w$ receives this request from some neighbor, it computes
the aggregated maximum value, according to whether $w \in \High$ for $\computeHigh$ (resp.\
$w\in\Low$ for $\computeLow$).
\STATE If node $w$ has received the request for the first time, $w$
forwards it to all neighbors (except $v$).
\STATE Once the entire network has been explored this way, i.e., the
request has been forwarded for $\textsc{diam}$ rounds, the aggregated
maximum values of the network size and $|\Low|$ (resp.\ $|\High|$) are
sent back to $u$, which receives them after $\le 2\textsc{diam}$
rounds.  
\end{algorithmic}
\caption{$\computeHigh$ and $\computeLow$.
}
\label{algo:computeHighLow}
\end{algorithm}
}

\onlyShort{\vspace{-0.1in}}
\subsection{Type-2 Recovery: Inflating and Deflating}  \label{sec:infldefl}
\onlyLong{%
We now describe an implementation of type-2 recovery that
yields amortized polylogarithmic bounds on messages and time.
We later extend these ideas\onlyShort{ below}\onlyLong{ (cf.\ Sec.~\ref{sec:improved})} to give $O(\log n)$
worst case bounds.}
Recall that we perform type-1 recovery in step $t$, as long as $\theta n$ nodes are in
$\High_{G_{t-1}}$ when a node is inserted, resp.\ in $\Low_{G_{t-1}}$,
upon a  deletion.



\onlyLong{
\begin{fact} \label{lem:inflateDeflate}
  If the algorithm performs type-2 recovery in $t$, the following holds:
  \begin{compactitem}
  \item[(a)] If a node is inserted in $t$, then $|\High_{G_{t-1}}|  < 
    \theta n$.
  \item[(b)] If a node is deleted in $t$, then $|\Low_{G_{t-1}}| < \theta 
    n$.
  \end{compactitem}
\end{fact}
}

\onlyLong{\subsubsection{Inflating the Virtual Graph}}
\onlyShort{\noindent\textbf{Inflating the Virtual Graph:}}
If node $v$ fails to find a spare node for a newly inserted neighbor and 
computes that $|\High_{G_{t-1}}|<\theta n$, i.e., only few nodes simulate
multiple virtual vertices each, it 
invokes Procedure $\inflate$\onlyShort{ (cf.\ Algorithm~\ref{algo:inflate}}\onlyLong{ (cf.\ Algorithm~\ref{algo:inflate} for the detailed pseudo code)}, which consists of two phases:

\onlyLong{\paragraph{Phase 1: Constructing a Larger $p$-Cycle}}
\onlyShort{\paragraph{Phase 1: Constructing a Larger $p$-Cycle}}
Node $v$ initiates replacing the current $p$-cycle $\Z_{t-1}(p_i)$
with the larger $p$-cycle $\Z_{t}(p_{i+1})$, for some prime number
$p_{i+1} \in (4p_i,8p_i)$.
This rebuilding request
is forwarded throughout the entire network to ensure that after this
step, every node uses the exact same new $p$-cycle $\Z_t$.
Intuitively speaking, each virtual vertex of $\Z_{t-1}$ is replaced by a
cloud of (at most $\zeta \le 8$) virtual vertices of $\Z_{t}$ and all
edges are updated such that $G_t$ is a virtual mapping of
$\Z_t$.

For simplicity, we use $x$ to denote both: an integer $x \in \mathbb{Z}_p$
and also the associated vertex in $V(\Z_t(p))$.
At the beginning of step $t$, all nodes are in agreement on the current virtual
graph $\Z_{t-1}(p_i)$, in particular, every node knows the prime number $p_i$.
To get a larger $p$-cycle, all nodes deterministically compute the (same)
smallest prime number $p_{i+1}  \in (4p_i,8p_i)$, i.e., $V(\Z_t(p_{i+1})) = 
\mathbb{Z}_{p_{i+1}}$.
(Local computation happens instantaneously and does not incur any
  cost (cf.\ Sec.\ \ref{sec: model}).)
Bertrand's postulate~\cite{Bertrand1850} states that for every $n>1$,
there is a prime between $n$ and $2n$, which ensures that $p_{i+1}$
exists.
Every node $u$ needs to determine the new set of vertices in
$\Z_{t}(p_{i+1})$ that it is going to simulate:
Let $\alpha = \frac{p_{i+1}}{p_i} \in O(1)$.
For every currently simulated vertex $x\in \Sim_{G_{t-1}}(u)$, node $u$ computes the constant 
\onlyLong{\begin{equation}\label{eq:cxDef}}
  \onlyShort{$}c(x) = \lfloor \alpha(x+1)\rfloor - \lfloor \alpha x\rfloor - 1,
  \onlyShort{$}
  \onlyLong{\end{equation}}
and replaces $x$ with the new virtual vertices $y_0,\dots,y_{c(x)}$ where
\onlyLong{\begin{equation} \label{eq:cx}}
\onlyShort{$}
  \text{$y_j = (\lfloor \alpha x \rfloor + j) \mod p_{i+1}$, for $0 \le j \le
c(x)$.}
\onlyShort{$}%
\onlyLong{\end{equation}}
Note that the vertices $y_0,\dots,y_{c(x)}$ form a cloud (cf.\ Sec.~\ref{subsec:
virtualexpander}) where the maximum cloud size is $\zeta\le 8$.
This ensures that the new virtual vertex set is a bijective mapping of
$\mathbb{Z}_{p_{i+1}}$.

Next, we describe how we find the edges of $\Z_t(p_{i+1})$:
First, we add new cycle edges (i.e. edges between $x$ and $x+1\mod
p_{i+1}$), which can be done in constant time by
using the cycle edges of the previous virtual graph $\Z_{t-1}(p_i)$.
For every $x$ that $u$ simulates, we need to add an edge to the 
node that simulates vertex $x^{-1}$.
Since this needs to be done by the respective simulating node of every
virtual vertex, this corresponds to solving a permutation routing instance.
Corollary 7.7.3 of \cite{S1998:Universal-abbrv}\onlyLong{ (cf.\ 
  Corollary~\ref{cor:routing})} states that, for any bounded
degree expander with $n$ nodes, $n$ packets, one per node, can be
routed (even online) according to an arbitrary permutation in
$O(\frac{\log n (\log\log n)^2}{\log\log\log n})$ rounds w.h.p.
Note that every node in the network knows the exact topology of the
current virtual graph (but not necessarily of the network graph $G_t$), 
and can
hence calculate all routing paths in this graph, which map to paths in the
actual network (cf.\ Fact~\ref{fact:distances}).
Since every node simulates a constant number of vertices, we can find the route
to the respective inverse by solving a constant number of permutation routing
instances.
\onlyLong{%
  The following lemma follows from the previous discussion\onlyShort{ (see the full paper \cite{full} for a
more detailed proof)}:

\begin{lemma} \label{lem:larger}
Consider a $t\ge 1$ where some node performs type-2 recovery via $\inflate$.
If the network graph $G_{t-1}$ is a $C$-balanced image of $\Z_{t-1}(p_i)$,
then Phase~1 of $\inflate$ ensures that every node computes the same virtual
graph in $O(\log n (\log\log n)^2)$ rounds such that the following hold:
\onlyShort{%
(a) $p_{i+1} = |\Z_{t}(p_{i+1})| \in (4p_i,8p_i)$, the 
  network graph is $(C\zeta)$-balanced, and the maximum clouds size is 
  $\zeta \le 8$.
(b) There is a bijective map between
  $\mathbb{Z}_{p_{i+1}}$ and $V(\Z_{t}(p_{i+1}))$.
(c) The edges of $\Z_{t}(p_{i+1})$ adhere to
  Definition~\ref{def:prime}.}
\onlyLong{%
\begin{compactenum}
\item[(a)] $p_{i+1} = |\Z_{t}(p_{i+1})| \in (4p_i,8p_i)$, the 
  network graph is $(C\zeta)$-balanced, and the maximum clouds size is 
  $\zeta \le 8$.
\item[(b)] There is a bijective map between
  $\mathbb{Z}_{p_{i+1}}$ and $V(\Z_{t}(p_{i+1}))$.
\item[(c)] The edges of $\Z_{t}(p_{i+1})$ adhere to
  Definition~\ref{def:prime}. 
\end{compactenum}
}
\end{lemma}
\onlyLong{
\begin{proof}
Property~(a) follows from the previous discussion.
For Property~(b), we first show set equivalence.
Consider any $z \in \mathbb{Z}_{p_{i+1}}$ and assume in contradiction that
$z \notin V(\Z_t(p_{i+1}))$.
Let $\alpha = {p_{i+1}} / {p_i}$ and let $x$ be the greatest integer
such that $z = \lfloor \alpha x \rfloor + k$, for some integer $k \ge 0$.
If $k \ge \alpha$, then 
$$z = \lfloor \alpha x + k \rfloor \ge \lfloor \alpha x + \alpha \rfloor = \lfloor \alpha(x+1) \rfloor,$$
which contradicts the maximality of $x$,
therefore, we have that  $k <  \alpha $.
It cannot be that $x < p_i$, since otherwise $z \in V(Z(p_{i+1}))$
according to \eqref{eq:cx}, which shows that $x \ge p_i$.
This means that 
\[
  z = \lfloor \alpha x \rfloor + k \ge \lfloor \alpha p_i \rfloor + k = \lfloor p_{i+1} \rfloor + k \ge p_{i+1},
\]
which contradicts $z \in \mathbb{Z}_{p_{i+1}}$, thus we have shown 
$\mathbb{Z}_{p_{i+1}} \subseteq V(\Z_t(p_{i+1}))$.
The opposite relation, i.e.\ $V(\Z_t(p_{i+1})) \subseteq \mathbb{Z}_{p_{i+1}}$,
is immediate since the values associated to vertices of $\Z_t(p_{i+1})$ are
computed modulo $p_{i+1}$.

To complete the proof of (b), we need to show that no two distinct vertices in
$V(\Z_t(p_{i+1}))$ correspond to the same value in
$\mathbb{Z}_{p_{i+1}}$, i.e., $V(\Z_t(p_{i+1}))$ is not a multi-set.
Suppose, for the sake of a contradiction, that there are $y = (\lfloor
\alpha x \rfloor + k) \mod\ p_{i+1}$ and $y' = (\lfloor \alpha x' \rfloor
+ k') \mod\ p_{i+1}$ with $y = y'$. 
By \eqref{eq:cx}, we know that $k' \le c(x)$,
hence to bound $k'$ it is sufficient to show that $c(x) <  \alpha$:
By \eqref{eq:cxDef}, we have that
$$
c(x) = \lfloor \alpha x+\alpha  - (\lfloor \alpha x \rfloor + 1)\rfloor 
 < \lfloor \alpha x + \alpha - \alpha x \rfloor \le \alpha.
$$
Note that the same argument shows that $k \le \alpha$.
Thus it cannot be that $y' = \lfloor \alpha x \rfloor + k + m p_{i+1}$,
for some integer $m\ge 1$.
This means that $x \ne x'$; wlog assume that $x  > x'$.
As we have shown above, $k' \le c(x) < \alpha$, which implies that
\[ 
  y' = \lfloor \alpha x' \rfloor + k' < \lfloor \alpha(x'+1) \rfloor \le \lfloor \alpha x \rfloor \le y,
\]
yielding a contradiction to $y=y'$.

For property (c), observe that all new cycle edges (i.e., of the form
$(x,x\pm 1))$ of $\Z_t(p_{i+1})$ are between nodes that were already simulating
neighboring vertices of $\Z_{t-1}(p_{i})$, thus every node $u$ can add these 
edges
in constant time.
Finally, we argue that every node can efficiently find the inverse vertex for
its newly simulated vertices:
Corollary 7.7.3 of \cite{S1998:Universal-abbrv} states that for any bounded
degree expander with $n$ nodes, $n$ packets, one per processor, can be
routed (online) according to an arbitrary permutation in $T=O(\frac{\log n
(\log\log n)^2}{\log\log\log n})$ rounds w.h.p.
Note that every node in the network knows the exact topology of the current
virtual graph (nodes do {not} necessarily know the {network} graph $G_t$!), and
can hence calculate all routing paths, which map to paths in the actual network
(cf.\ Fact~\ref{fact:distances}).
Since every node simulates a constant number of vertices, we can find the route
to the respective inverse by performing a constant number of iterations of
permutation routing, each of which takes $T$ rounds.
\end{proof}
}
}

\onlyLong{\paragraph{Phase 2: Rebalancing the Load}}
\onlyShort{\paragraph{Phase 2: Rebalancing the Load}}
Once the new virtual graph $\Z_t(p_{i+1})$ is in place, each real node
simulates a greater number (by a factor of at most $\zeta$) of virtual
vertices and now a random walk is guaranteed to find a spare virtual vertex on the
first attempt with high probability, according to
Lemma~\ref{lem:gillman}.(a).
At the beginning of the step, the virtual mapping $\Phi_{t-1}$ was
$4\zeta$-balanced.
This, however, is not necessarily the case after Phase~1, i.e., replacing
$\Z_{t-1}$ by $\Z_t$.
A node could have been simulating
$4\zeta$ virtual vertices \emph{before} $\inflate$ was invoked and now
might be simulating $4\zeta^2$ vertices of $\Z_t(p_{i+1})$.
In fact, this can be the case for a $\theta$-fraction of the nodes.
To ensure a $4\zeta$-balanced mapping at the end of step $t$, we thus need
to rebalance  these additional vertices among the other (real) nodes.
Note that this is always possible, since $(1-\theta)n$ 
nodes had a load of $1$ before invoking $\inflate$ and simulate
only $\zeta$ virtual vertices each at the end of Phase~1.
A node $v$ that has a load of $k'>4\zeta$ vertices of $\Z_t(p_{i+1})$,
proceeds as follows, for each vertex $z$ of the (at most constant) vertices that
it needs to redistribute:
Node $v$ marks all of its vertices as \emph{full} and initiates a random
walk of length $\Theta(\log n)$ on the virtual graph $\Z_{t}(p_{i+1})$,
which is simulated on the actual network.
If the walk ends at a vertex $z'$ simulated at some node $w$ that is not
marked as \emph{full}, and no other random walk simultaneously ended up at
$z'$, then $v$ transfers $z$ to $w$.
This ensures that $z$ is now simulated at a node that had a load of
$<4\zeta$.
A node $w$ immediately marks all of its vertices as \emph{full}, once its load
reaches $2\zeta$.
Node $v$ repeatedly performs random walks until all of the $k'-4\zeta$
vertices are transfered to other nodes.


\newcommand{\lemAtMostInflate}{
  Suppose that $G_{t-1}$ is $4\zeta$-balanced and type-2 recovery is performed
  in $t$ via $\inflate$ or $\deflate$.
  The following holds:
  \onlyShort{%
  (a) $G_t$ is $4\zeta$-balanced.
  (b) With high probability, step $t$ completes in 
    $O(\log^3 n )$.
  (c) With high probability, nodes send $O(n\log^2 n)$ messages.
  (d) The number of topology changes is $O(n)$.  
  }
  \onlyLong{%
  \begin{compactitem}
  \item[(a)] $G_t$ is $4\zeta$-balanced.
  \item[(b)] With high probability, step $t$ completes in 
    $O(\log^3 n )$.
  \item[(c)] With high probability, nodes send $O(n\log^2 n)$ messages.
  \item[(d)] The number of topology changes is $O(n)$.
  \end{compactitem}
  }
}
\onlyLong{
\begin{lemma}[Simplified Type-2 recovery] \label{lem:atMostInflate}
  \lemAtMostInflate
\end{lemma}
\begin{proof}
  Here we will show the result for $\inflate$.
  In Sec.~\ref{sec:deflate}, we will argue the same properties for 
  $\deflate$ (described below).

  Property (d) follows readily from the description of Phase~1.
  For (a), we observe that, in Phase~1, $\inflate$ replaces each virtual
  vertex with a cloud of virtual vertices.
  Moreover, nodes only redistribute vertices such that their load does not
  exceed $4\zeta$.
  It follows that every node simulates at least one vertex, thus
  $\Phi_{t}$ is surjective.
  What remains to be shown is that every node has
  a load $\le 4\zeta$ at the end of $t$.

  Consider any node $u$ that has $\Load(v)\in(2\zeta,4\zeta)$ after
  Phase~1.
  To see that $u$'s load does not exceed $4\zeta$, recall that, according 
  the description of Phase~2, $u$ will mark all its vertices as 
  \emph{full} and henceforth will not accept any new vertices.
  By Fact~\ref{lem:inflateDeflate}.(a), at most $\theta n$ nodes have a 
  load $>1$ in $U_{t}$.
  Let $Balls_0$ be the set of vertices that need to be redistributed.
  Lemma~\ref{lem:larger}.(a) tells us that the every vertex in
  $\Z_{t-1}(p_i)$ is replaced by (at most) $\zeta$ new vertices in 
  $\Z_t(p_{i+1})$,
  which means that $|Balls_0| \le 4\theta(\zeta^2 - \zeta)n,$ since every 
  such high-load node continues to simulate $4\zeta$ vertices by itself.
  
  To ensure that this redistribution can be done in polylogarithmic time,
  we need to lower bound the total number of available places (i.e.\ the 
  bins) for these virtual vertices (i.e.\ the balls).
  By Fact~\ref{lem:inflateDeflate}.(a), we know that $\ge (1-\theta)n$ 
  nodes have a load of at most $\zeta$ after Phase~1.
  These nodes do not mark their vertices as \emph{full}, and thus accept 
  to simulate additional vertices until their respective load reaches 
  $2\zeta$.
  Let $Bins$ be the set of virtual vertices that are not marked as 
  \emph{full}; It holds that $|Bins| \ge (1-\theta)\zeta n.$

  We first show that with high probability, a constant fraction of random
  walks end up at vertices in $|Bins|$.
  Since $\Z_t(p_{i+1})$ is a regular expander, the distribution of the random 
  walk converges to the uniform distribution (e.g., \cite{MU2005:Probability}) 
  within $O(\log \sigma)$ random steps where $\sigma = |Z^{i+1}| \in 
  \Theta(n)$.
  More specifically, the distance (measured in the maximum norm) to the 
  uniform distribution, represented by a vector 
  $(1/\sigma,\dots,1/\sigma)$, can be bounded by
  $\frac{1}{100\sigma}$.
  Therefore, the probability for a random walk token to end up at a 
  specific vertex is within $[\frac{99}{100\sigma},\frac{101}{100\sigma}]$.
  Recall that, after Phase~1 all nodes have computed the same graph 
  $\Z_t(p_{i+1})$ and thus use the same value $\sigma$.

  We divide the random walks into \emph{epochs} where an epoch is the 
  smallest interval of rounds containing $c\log n$ random walks.
  We denote the number of vertices that still need to be redistributed
  at the beginning of epoch $i$ as $Balls_i$. 
  \onlyShort{
    We first bound the number of rounds taken by an epoch (cf.\ full paper \cite{full} for the
  proof):}
  \begin{claim} \label{claim:balls}
  Consider a fixed constant $c$.
  If $|Balls_i| \ge c\log n$, then epoch $i$ takes $O(\log^2 n)$
  rounds, w.h.p.
  Otherwise, if $|Balls_j| < c\log n$, then $j$ comprises 
  $O(\log^3 n)$ rounds w.h.p.
  \end{claim}

  \onlyLong{
  \begin{proof}
  We will now show that an epoch lasts at most $O(\log^3 n)$ rounds with
  high probability.
  First, suppose that $|Balls_i| \ge c\log n$.
  By Lemma~\ref{lem:randomWalk}, we know that even a linear number of
  parallel walks (each of length $\Theta(\log n)$) will complete within
  $O(\log^2 n)$ rounds w.h.p.
  Therefore, epoch $i$ consists of $O(\log^2 n)$ rounds, since
  $\Omega(\log n)$ random walks are performed in parallel.
  In the case where $|Balls_j| <c\log n$, it is possible that an epoch
  consists of random walks that are mostly performed
  sequentially by the same nodes. 
  Thus we add a $\log n$ factor to ensure that epoch $j$ consists of 
  $c\log n$ walks.
  By Lemma~\ref{lem:randomWalk} we get a bound of $O(\log^3 n)$ rounds.
  \end{proof}
  } 
  Next, we will argue that
  after $O(\log n)$ epochs, we have $|Balls_j| <c\log n$.
  Thus consider any epoch $i$ where $|Balls_i| \ge c\log n$.
  We bound the probability of the indicator
  random variable $Y_k$ that is $1$ iff the walk associated with the
  $k$-th vertex ends up at a vertex that was already marked \emph{full} when the
  walk was initiated.
  (In particular, $Y_k=0$ if the $k$-th walk ends up at $z$ and $z$ became
  \emph{full} in the current iteration but was not marked \emph{full} before.)
  Note that all $Y_k$ are independent.
While the number of available bins (i.e.\ non-full vertices) will decrease 
over time, we know from \eqref{eq:theta} that $|Bins| -  |Balls_0| >
  \frac{9}{10}|Bins|$; thus, at any epoch,  we can use the bound
    $|Bins| \ge ({9}/{10})(1-\theta)\zeta n.$
  This shows that
    $\Prob{Y_k=1} \le \frac{101}{100\sigma}\left(\sigma - |Bins|\right)
           \le \frac{101}{100}\left(1-\frac{9(1-\theta)\zeta 
              n}{10\sigma}\right).$
  From $\sigma \le \zeta(1-\theta)n + 4\zeta^2\theta n$
  and the fact that \eqref{eq:theta} implies  $ 1-\frac{9 (1-\theta) 
    \zeta}{10 ((1-\theta) \zeta+4 \zeta^2 \theta)}<3/20,
  $ we get that
  $\Prob{Y_k=1} \le ({101}/{100}) \cdot ({3}/{20})$.
  Let $Y = \sum_{k\in Balls_i} Y_k$.  Since $|Balls_i|=\Omega(\log n)$ , 
  we can use
  a Chernoff bound (e.g.\ \cite{MU2005:Probability}) to show that
  $
\Pr\left[Y \ge ({909}/{1000})|Balls_i| \ge 6 \E[Y]\right] \le 
2^{-\frac{909}{1000}|Balls_i|},
  $
  thus with high probability (in $n$), a constant fraction of the
  random walks in epoch $i$ will end up at non-full vertices.
  We call these walks \emph{good balls} and denote this set as $Good_i$.

  We will now show that a constant fraction of good balls do not end up at the 
  same bin with high probability,
  i.e., we are able to successfully redistribute the associated vertices in 
  this epoch.
  Let $X_k$ be the indicator random variable that is $1$ iff the $k$-th
  ball is eliminated. 
  We have
  $\Pr[X_k=1] \ge (1-\frac{101}{100|Bins|})^{|Good_i|-1} \ge e^{-\Theta(1)}$,
  i.e., at least a constant fraction of the balls in $Good_i$ are
  eliminated on expectation.

  Let $W$ denote the number of eliminated vertices in epoch $i$, which is 
  a function $f(B_1,\dots,B_{|Good_i|})$ where $B_j$ denotes the bin 
  chosen by the $j$-th ball.
  Observe that changing the bin of some ball can affect the elimination of
  at most one other ball.
  In other words, $W$ satisfies the Lipschitz condition and we
  can apply the method of bounded differences.
  By the Azuma-Hoeffding Inequality (cf.\ Theorem~12.6 in
  \cite{MU2005:Probability}), we get a sharp concentration bound for $W$, 
  i.e., with high probability, a constant fraction of the balls are 
  eliminated in
  every epoch.

  We have therefore shown that after $O(\log n)$ epochs, we are left with
  less than
  $c\log n$ vertices that need to be redistributed, w.h.p.
  Let $j$ be the first epoch when $|Balls_j| <c\log n$.
  Note that epoch $j$ consists of $\Omega(\log n)$ random walks where some 
  nodes perform multiple random walks.
  By the same argument as above, we can show that with high probability, a
  constant fraction of these walks will end up at some
  non-full vertices without conflicting with another walk and are thus
  eliminated. Since we only need $c\log n$ walks to succeed, this
  ensures that the entire set $Balls_j$ is redistributed
  w.h.p.\ by the end of epoch $j$, which shows (a).

  By Claim~\ref{claim:balls}, the first $O(\log n)$ epochs can each last
  $O(\log^2 n)$ rounds, while only epoch $j$ takes $O(\log^3 n)$ rounds.
  Altogether, this gives a running time bound of $O(\log^3 n)$, as
  required for (b).
  For Property (c), note that the flooding of the inflation request to all
  nodes in the network requires $O(n)$ messages.
  This, however, is dominated by the time it takes to redistribute the 
  load: each epoch might use $O(n\log n)$ messages.
  Since we are done w.h.p.\ in $O(\log n)$ epochs, we get a total message
  complexity of $O(n \log^2 n)$.
  \onlyLong{
    For (d), observe that the sizes of the virtual expanders $\Z_{t-1}(p_i)$ 
    and $\Z_t(p_{i+1})$ are both in $O(n)$. Due to their constant degrees, at 
    most
  $O(n)$ edges are affected by replacing the edges of $\Z_{t-1}(p_i)$ with 
  the ones of $\Z_t(p_{i+1}$, yielding a total of $O(n)$ topology changes 
  for)
  $\inflate$.
} \qed
\end{proof}
}

%

\onlyLong{\subsubsection{Deflating the Virtual Graph} \label{sec:deflate}}
\onlyShort{\noindent\textbf{Deflating the Virtual Graph}:}
When the load of all but $\theta n$ nodes exceeds $2\zeta$ and some node $u$ is
deleted, the high probability bound of Lemma~\ref{lem:gillman} for the random
walk invoked by neighbor $v$ no longer applies. 
In that case, node $v$ invokes Procedure $\deflate$ to reduce the overall load\onlyShort{ (cf.\ full paper \cite{full}.}\onlyLong{ (cf.\ Algorithm~\ref{algo:deflate}).}
Analogously as $\inflate$, Procedure $\deflate$ consists of two phases:

\onlyLong{\paragraph{Phase 1: Constructing a Smaller $p$-Cycle}}
\onlyShort{\paragraph{Phase 1: Constructing a Smaller $p$-Cycle}}
To reduce the load of simulated vertices, we replace the current
$p$-cycle $Z_{t-1}(p_i)$ with a smaller $p$-cycle $\Z_t(p_s)$ where $p_s$ 
is a
prime number in the range $(p_i/8,p_i/4)$.

Let $\alpha = {p_i}/{p_s}$.
Any virtual vertex $x \in \Z_{t-1}(p_i)$, is (surjectively) mapped to some $y_x
\in \Z_t(p_{s})$ where $y = \lfloor{x}/{\alpha} \rfloor$.
Note that we only add $y$ to $V(\Z_t(p_s))$ if there is no smaller $x' \in
\Z_{t-1}(p_i)$ that yields the same $y$.
This mapping guarantees that, for any element in $\mathbb{Z}_{p_s}$, we have
exactly $1$ virtual vertex in $\Z_{t}(p_s)$:
Suppose that there is some $y \in \mathbb{Z}_{p_s}$ that is not hit by our mapping,
i.e., for all $x \in \mathbb{Z}_{p_i}$, we have $y > \lfloor
\frac{x}{\alpha} \rfloor$.
Let $x'$ be the smallest integer  such that $y = \lfloor
\frac{x'}{\alpha} \rfloor$.
For such an $x'$, it must hold that $\alpha y \le x' < \alpha(y+1)$.
Since $\alpha>1$, clearly $x'$ exists.
By assumption, we have $x' \ge p_i$, which yields
$  \left\lfloor {p_i}/{\alpha} \right\rfloor \le \left\lfloor
{x'}/{\alpha} \right\rfloor = y < p_s.$
Since $p_s = {p_i}/\alpha$, we get
$ \left\lfloor p_s\right \rfloor < p_s,$
which is a contradiction to $p_s \in \mathbb{N}$.
Therefore, we have shown that $\mathbb{Z}_{s}\subseteq V(\Z_t(p_s))$.
The opposite set inclusion can be shown similarly.

\onlyLong{
For computing the edges of $\Z_t(p_s)$, note that any cycle edge $(y,y\pm 1) \in
E(\Z_t(p_s))$, is between nodes $u$ and $v$ that were at most $\alpha$ hops
apart in $G_t$,
since their distance is at most $\alpha$ in the virtual graph $\Z_{t-1}(p_i)$.
Thus any such edge can be added by exploring a neighborhood of constant-size in
$O(1)$ rounds via the cycle edges (of the current virtual graph) $\Z_{t-1}(p_i)$
in $G_t$.
To add the edge between $y$ and its inverse $y^{-1}$, we
proceed along the lines of Phase~1 of $\inflate$, i.e., we solve
permutation routing on $\Z_{t-1}(p_i)$, taking $O(\frac{\log n (\log\log
n)^2}{\log\log\log n})$
rounds.
}
\onlyShort{
To add cycle edges and the edge between $y$ and $y^{-1}$, we
proceed as in Phase~1 of $\inflate$, i.e., we solve
permutation routing on $\Z_{t-1}(p_i)$, taking $O(\frac{\log n (\log\log
n)^2}{\log\log\log n})$
rounds.
}
\onlyLong{%
The following lemma summarizes the properties of Phase~1:
\begin{lemma} \label{lem:smaller}
If the network graph $G_{t-1}$ is a balanced map of $\Z_{t-1}(p_i)$, then
Phase~1 of $\deflate$ ensures that every node computes the same virtual
graph $\Z_t(p_s)$ in $O(\log n (\log\log n)^2)$ rounds such that 
\begin{compactenum}
\item[(a)] $p_s = |\Z_t(p_s)| \in (p_i/8,p_i/4)$, for some prime $p_s$;
\item[(b)] there is a one-to-one mapping between
  $\mathbb{Z}_{p_s}$ and $V(\Z_t(p_s))$;
\item[(c)] the edges of $\Z_t(p_s)$ adhere to Definition~\ref{def:prime}. 
\end{compactenum}
\end{lemma}
\begin{proof}
Property (a) trivially holds.
For (b), observe that by description Phase~1, we map $x \in
\Z_{t-1}(p_i)$ surjectively to
$y_x \in \Z_t(p_{s})$ using the mapping $y_x = \lfloor\frac{x}{\alpha}
\rfloor$ where $\alpha = \frac{p_i}{p_s}$.
Note that we only add $y_x$ to $V(\Z_t(p_s))$ if there is no smaller $x \in 
\Z_{t-1}(p_i)$ that yields the same value in $\mathbb{Z}_{p_s}$, which guarantees that
$V(\Z_t(p_s))$ is not a multiset.
Suppose that there is some $y \in \mathbb{Z}_{p_s}$ that is not hit by our mapping,
i.e., for all $x \in \mathbb{Z}_{p_i}$, we have $y > \lfloor
\frac{x}{\alpha} \rfloor$.
Let $x'$ be the smallest integer  such that $y = \lfloor
\frac{x'}{\alpha} \rfloor$.
For such an $x'$, it must hold that $\alpha y \le x' < \alpha(y+1)$.
Since $\alpha>1$, clearly $x'$ exists.
By assumption we have $x' \ge p_i$, which yields
\[
  \left\lfloor \frac{p_i}{\alpha} \right\rfloor \le \left\lfloor \frac{x'}{\alpha} \right\rfloor < p_s.
\]
Since $\alpha = \frac{p_i}{p_s}$, we get
\[
 \left\lfloor p_s\right \rfloor = \left\lfloor \frac{p_i}{\alpha} \right\rfloor < p_s,
\]
which is a contradiction to $p_s \in \mathbb{N}$.
Therefore, we have shown that $\mathbb{Z}_{s}\subseteq V(\Z_t(p_s))$.
To see that $V(\Z_t(p_s))\subseteq \mathbb{Z}_{s}$, suppose that we add a 
vertex
$y \ge p_s$ to $V(\Z_t(p_s))$.
By the description of Phase~1, this means that there is an $x \in 
V(\Z_{t-1}(p_{i}))$,
i.e., $x \le p_i - 1$, such that $y = \lfloor \frac{x}{\alpha} \rfloor$.
Substituting for $\alpha$ yields a contradiction to $y\ge p_s$, since
\[
y = \left\lfloor \frac{x}{\alpha} \right\rfloor \le \left\lfloor\frac{p_{i}-1}{\alpha} \right\rfloor = \left\lfloor p_s - \frac{p_s}{p_i}\right\rfloor < p_s.
\]

For property (c), note that any cycle edge $(y,y\pm 1) \in E(\Z_t(p_s))$, is
between nodes $u$ and $v$ that were at most $\alpha$ hops apart in $G_t$,
since their distance can be at most $\alpha$ in $\Z_{t-1}(p_i)$. Thus any such 
edge
can be added by exploring a neighborhood of constant-size in $O(1)$ rounds
via the cycle edges of $\Z_{t-1}(p_i)$ in $G_t$.
To add an edge between $y$ and its inverse $y^{-1}$, we proceed along the
lines of the proof of Lemma~\ref{lem:larger}, i.e., we solve permutation
routing on $\Z_{t-1}(p_i)$, taking $O(\frac{\log n (\log\log 
  n)^2}{\log\log\log n})$ rounds.
\end{proof}
}

\onlyLong{\paragraph{Phase 2: Ensuring a Virtual Mapping}}
\onlyShort{\paragraph{Phase 2: Ensuring a Virtual Mapping}}
After Phase~1 is complete, the replacement of multiple virtual vertices in
$\Z_{t-1}(p_i)$ by a single vertex in $\Z_t(p_s)$, might lead to the case where
some nodes are no longer simulating \emph{any} virtual vertices.
A node that currently does not simulate a vertex, marks itself as
\emph{contending} and repeatedly keeps initiating random walks on $\Z_t(p_s)$
(that are simulated on the actual network graph) to find spare vertices.
\onlyLong{
Moreover, a node $w$ that does simulate vertices, marks an arbitrary 
vertex as \emph{taken} and transfers its other vertices to other nodes if 
requested.
To ensure a valid mapping $\Phi_t$, we need to transfer non-taken
vertices to contending nodes if the random walk of a contending node hits a
non-taken vertex $z$ and no other walk ends up at $z$ simultaneously.
}
\onlyLong{A similar analysis as for Phase 2 of $\inflate$ shows 
  Lemma~\ref{lem:atMostInflate} for deflation steps.}

\onlyLong{
Lemmas~\ref{lem:atMostSimple} and \ref{lem:atMostInflate}
imply the following: 
\begin{lemma} \label{lem:mapping}
At any step $t$, the network graph $G_t$, is $4\zeta$-balanced, i.e.,
$G_t$ has constant node degree and $\lambda_{G_t} \le \lambda$ where $1-\lambda$
is the spectral gap of the $p$-cycle expander family.
\end{lemma}
\begin{proof}
The result follows by induction on $t$.
For the base case, note that we initialize $G_0$ to be a
virtual mapping of the expander $\Z_0(p_0)$, which obviously guarantees that 
the network is $4\zeta$-balanced.
For the induction step, we perform a case distinction depending on whether 
$t$ is a simple or inflation/deflation step and apply the respective 
result, i.e.\ Lemmas~\ref{lem:atMostSimple} or
\ref{lem:atMostInflate}.
\end{proof}
}


\onlyLong{
\subsection{Amortizing (Simplified) Type-2 Recovery}  \label{sec:amortized}
\onlyShort{
  \noindent\textbf{Amortizing (Simplified) Type-2 Recovery:}  
}
We will now show that the expensive inflation/deflation steps occur rather
infrequently. This will allow us to amortize the cost of the worst case
bounds derived in Section~\ref{sec:infldefl}. 
Suppose that step $t$ was an inflation step.
By Fact~\ref{lem:inflateDeflate}.(a), this means that at least $(1-\theta)n$ nodes
had a load of $1$ at the beginning of $t$, and thus a load of $\le \zeta$ at the end of
$t$.
Thus, even after redistributing the additional load of the $\theta n$ nodes that
might have had a load of $>4\zeta$, a large fraction of nodes are in $\Low$ and
$\High$ at the end of $t$.
This guarantees that we perform type-1 recovery in $\Omega(n)$ steps, before the
next inflation/deflation is carried out\onlyShort{(cf.\ full paper \cite{full} for the proof)}.
A similar argument applies to the case when $\deflate$ is invoked, thus 
yielding amortized polylogarithmic bounds on messages and rounds per every 
step.
}

\onlyLong{
\begin{lemma} \label{lem:inflateDeflateSpacing}
  There exists a constant $\delta$ such that the following holds:
  If $t_1$ and $t_2$ are steps where type-2 recovery is performed (via
  $\inflate$ or $\deflate$),
  then $t_{1}$ and $t_{2}$ are separated by at 
  least $\delta n \in \Omega(n)$ steps with type-1 recovery where $n$ is the size of
  $G_{t_1}$.
 \end{lemma}
For the proof of Lemma~\ref{lem:inflateDeflateSpacing} we require the 
following 2 technical results:
\begin{claim} \label{claim:inflateLow}
  Suppose that $t$ is an inflation step. Then $|\Low_{t}| \ge (\theta + 
  \frac{1}{2})n$.
\end{claim}
\begin{proof}[of Claim~\ref{claim:inflateLow}]
  First, consider the set of nodes $S=U_t\setminus\High_{U_t}$, i.e.,  
  $\Load_{U_t}(u)=1$ for all $u \in S$. By 
  Fact~\ref{lem:inflateDeflate}.(a), we have $|S| \ge (1-\theta)n$.  
  Clearly, any such node $u \in S$ simulates at most $\zeta$ virtual 
  vertices after generating its own vertices for the new virtual graph, 
  hence the only way for $u$ to reach $\Load_{t}(u)>2\zeta$ is by taking 
  over vertices generated by other nodes. By the description of procedure 
  $\inflate$, only (a subset of) the nodes in $\High_{U_t}$ redistribute 
  their load by performing random walks. 
  By Lemma~\ref{lem:mapping}, we can assume that $G_{t_1-1}$ is 
  $4\zeta$-balanced.
  Since $|\High_{U_t}|<\theta n$, we have a total of $\le (4\zeta -4) 
  \theta n$ clouds that need to be redistributed. Observe that $v$ 
  continues to simulate $4$ clouds (i.e.\ $4\zeta$ nodes) by itself. Since 
  every node that is in $S$, has at most $\zeta$ virtual nodes, we can 
  bound the size of $\Low_{t}$ by subtracting the redistributed clouds 
  from $|S|$.  For the result to hold we need to show that
  \[(\theta + {1} / {2})  \le 1-\theta - (4\zeta - 4)\theta,\]
  which immediately follows by Inequality~\eqref{eq:theta}.
\end{proof}
\begin{claim} \label{claim:deflateHigh}
  Suppose that $t$ is a deflation step. Then $|\High_{t}| \ge 
  (\theta+\frac{1}{4\zeta})n$.
\end{claim}
\begin{proof}[of Claim~\ref{claim:deflateHigh}]
  Consider the set $S=\{u\colon \Load_{U_t}(u)>2\zeta\}$. Since $S=U_t 
  \setminus \Low_{U_t}$, Fact~\ref{lem:inflateDeflate}.(b) tells us that 
  $|S|\ge (1-\theta)n$ and therefore we have a total load of least 
  $(1-\theta)(2\zeta+1)n + \theta n$ in $U_t$.  By description of 
  procedure $\deflate$, every cloud of virtual vertices is contracted to a 
  single virtual vertex.  After deflating we are left with
  \begin{equation*} 
    \Load(G_t) \ge \left((1-\theta)(2+\frac{1}{\zeta}) + 
    \frac{\theta}{\zeta}\right)n.
  \end{equation*}
  To guarantee the sought bound on $\High_{t}$, we need to show that 
  $\Load(G_t) \ge (1 + \theta + \frac{1}{4\zeta})n$. This is true, since by 
  \eqref{eq:theta} we have $\theta \le \frac{1}{3}+\frac{1}{4\zeta}$.  
  Therefore, by the pigeon hole principle, at least $\theta+\frac{1}{4\zeta}$ 
  nodes have a load of at least $2$.
\end{proof}

\begin{proof}[Proof of Lemma~\ref{lem:inflateDeflateSpacing}]
  It is easy to see that the values computed by procedures $\computeHigh$ 
  and $\computeLow$ cannot simultaneously satisfy the thresholds of 
  Fact~\ref{lem:inflateDeflate}, i.e., $\inflate$ and $\deflate$ are 
  never called in the same step.  
  Let ${t_1,t_2,\dots}$ be the set of steps where, for every $i\ge 1$, a 
  node calls either Procedure $\inflate$ or Procedure $\deflate$ in $t_i$.  
  Fixing a constant $\delta$ such that 
  \begin{equation} \label{eq:delta}
    \delta \le {1} / {4\zeta}, 
  \end{equation} we need to show that
  $t_{i+1} - t_{i} \ge \delta n$.

  We distinguish several cases:
  \begin{description}
  \item[1. $t_i$ \inflate; $t_{i+1}$ \inflate:]
  By Fact~\ref{lem:inflateDeflate}.(a) we know that
  $\High_{U_{t_i}}$ contains less than $\theta n$ nodes. Since we inflate 
  in $t_i$, every node generates a new cloud of virtual vertices, i.e., 
  the load of every node in $U_{t_i}$ is (temporarily) at least $\zeta$ 
  (cf.\ Phase~1 of $\inflate$).  Moreover, the only way that 
  the load of a node $u$ can be reduced in $t_i$, is by transferring some 
  virtual vertices from $u$ to a newly inserted node $w$.  However, by the 
  description of $\inflate$ and the assumption that $\zeta>2$, we still 
  have $\Load_{t}(u)>1$ (and $\Load_{t}(w)\ge 1$), and therefore
  $\High_{G_{t_{i}}}\supseteq V(G_{t_i})\setminus\{w\}$.  Since the virtual 
  graph (and hence the total load) remains the same during the interval  
  $(t_i,t_{i+1})$, it follows by  Lemma~\ref{lem:mapping} that $\High$ can 
shrink by at most the number of insertions during $(t_i,t_{i+1})$.  Since 
$|\High_{U_{t_{i+1}}}|<\theta n$, more than $(1-\theta)n-1 > \delta n$ 
insertions are necessary.

  \item[2. $t_i$ \deflate; $t_{i+1}$ \deflate:]
  We first give a lower bound on the size of $\Low_{G_{t_i}}$. By 
  Lemma~\ref{lem:atMostInflate}, we know that load at every node is at 
  most $4\zeta$ in $U_{t_i}$. Since every virtual cloud (of size $\zeta$) 
  is contracted to a single virtual zertex in the new virtual graph, the 
  load at every node is reduced to at most $4$. Clearly, the nodes that 
  are redistributed do not increase the load of any node beyond $4$, thus 
  $\Low_{t}=G_t$. Analogously to Case 1, the virtual graph is not 
  changed until $t_{i+1}$ and Lemma~\ref{lem:mapping} tells us that $\Low$ 
  is only affected by deletions, i.e., $(1-\theta)n \ge \delta n$ steps 
  are necessary before step $t_{i+1}$.
    
  \item[3. $t_i$ \inflate; $t_{i+1}$ \deflate:]
  By Claim~\ref{claim:inflateLow}, we have $|\Low_{G_{t_i}}| \ge (\theta + 
  1/2) n$, while Fact~\ref{lem:inflateDeflate}.(b) tells us that 
  $|\Low_{G_{t_{i+1}}}|<\theta n$. Again, Lemma~\ref{lem:mapping} implies 
  that the adversary must delete at least $n/2 \ge \delta n$
  nodes during $(t_i,t_{i+1}]$. 

  \item[4. $t_i$ \deflate; $t_{i+1}$ \inflate:]
  By Claim~\ref{claim:deflateHigh}, we have $|\High_{G_{t_i}}| \ge (\theta 
  + \frac{1}{4\zeta})n$, and by Fact~\ref{lem:inflateDeflate}.(a), we 
  know that $|\High_{G_{t_{i+1}}}|<\theta n$.
  Applying Lemma~\ref{lem:mapping} shows that we must have more than 
  $\frac{1}{4\zeta}n \ge \delta n$ deletions before $t_{i+1}$.  
  \end{description}
\end{proof}
The following corollary summarizes the bounds that we get when using the simplified type-2 recovery:\footnote{We will show in Sec~\ref{sec:improved} how to get \emph{worst case} $O(\log n)$ complexity bounds.} 
\begin{corollary} \label{cor:runtime}
  Consider the (simplified) variant of \dex\ that uses Procedures~\ref{algo:inflate} and \ref{algo:deflate} to handle type-2 recovery.
  With high probability, the amortized running time of any step is
  $O(\log n )$ rounds, the amortized message complexity of any recovery
  step is $O(\log^2 n)$, while the amortized number of topology changes is $O(1)$.
\end{corollary}
}

\onlyShort{
\begin{algorithm}[h!]
\begin{algorithmic}[1]
  \small
\item[] \textbf{Given:} current network size $n$ (as computed
  by $\computeHigh$). All virtual vertices and nodes are unmarked.
\item[]
\item[] \textbf{\boldmath Phase 1. Compute larger $p$-cycle:}
\STATE Inserted node $u$ forwards an inflation request through the entire network. 
\STATE Initiating node $u$ floods a request to all other nodes to run this
process simultaneously; takes $O(\log n)$ time. 
\STATE Since every node $u$ knows the same virtual graph $\Z_{t-1}(p_i)$, all nodes
locally compute the same prime $p_{i+1} \in (4p_i,8p_i)$ and therefore the
same virtual expander $\Z_t(p_{i+1})$ with vertex set $\mathbb{Z}_{p_{i+1}}$.
\STATE (Compute the new set of locally simulated virtual vertices.) \\
Let $\alpha = \frac{p_{i+1}}{p_i}$ and define the function 
\begin{equation} \label{eq:cxInflate}
c(x) = \lfloor \alpha(x+1)\rfloor - \lfloor \alpha x\rfloor - 1.
\end{equation}
Replace every $x \in \Sim(u)$ (i.e.\ $x\in \Z_{t-1}(p_{i})$) with a cloud of virtual vertices $y_0,\dots,y_{c(x)}$ where 
$y_k = (\lfloor \alpha x \rfloor + k) \mod\ p_{i+1}$, for $0 \le k \le c(x)$.
That is, $\cloud{y_0}=\cdots=\cloud{y_{c(x)}}=\{y_0,\dots,y_{c(x)}\}$.

\FOR{every $x \in \Sim(u)$ and every $y_k$, $(0\le k\le c(x))$} 
\item[] (Compute the new set of edges.)\\ 
\begin{compactdesc}
\item Cycle edges: Add an edge between $u$ and the nodes $v$ and $v'$
  that simulate
  $y_{k}-1$ and $y_{k}+1$ by using the cycle edges of $\Z_{t-1}(p_i)$ in $G_t$.
\item Inverse edges: Add an edge between $u$ and the
node $v$ that simulates $y_k^{-1}$; node $v$ is found by 
solving a permutation routing instance.
\end{compactdesc}
\ENDFOR

\STATE After the construction of $\Z_t(p_{i+1})$ is complete, we
transfer a (newly generated) virtual vertex to the inserted node $u$ from
its neighbor $v$.
\item[]
\item[] \textbf{Phase 2. Perform load balancing:}
\IF{a node $w$ has $\Load(w)>2\zeta$ (i.e.\ $w\notin\Low$)}
\STATE Node $w$ marks all vertices in $\Sim(w)$ as \emph{full}.
\ENDIF
\IF{a node $v$ has load $k' > 4\zeta$ vertices}
  \item[] (Distribute all except $4\zeta$ vertices to other nodes.) 
  \FOR{each of the $k'-4\zeta$ vertices}
  \STATE Node $v$ marks itself as \emph{contending}.
  \WHILE{$v$ is \emph{contending}}
    \STATE Every \emph{contending} node $v$ performs a random walk of
    length $T=\Theta(\log n)$ on the virtual graph $\Z_t(p_{i+1})$ by
    forwarding a token $\tau_v$.
    This walk is simulated on the actual network $U_t$ (with constant
    overhead).
    To account for congestion, we give this walk $\rho=O(\log^2 n)$
    rounds to complete; once a token has taken $T$ steps it remains at its
    current vertex.
    \STATE If, after $\rho$ rounds, $\tau_v$ has reached a virtual
    vertex $z$ (simulated at some node $w$), no other
    token is currently at $z$, and $z$ is not marked as $\emph{full}$, 
    then $v$ marks itself as
    \emph{non-contending} and transfers a virtual vertex to $w$.
    Moreover, if the new load of $w$ is $>2\zeta$, we mark all vertices at
    $w$ as \emph{full}. \label{line:inflateTransfer}
  \ENDWHILE
  \ENDFOR
\ENDIF
\end{algorithmic}
\caption{\small Procedure \inflate. This is the simplified inflation yielding amortized $O(\log n)$ bounds. In the full paper we provide the pseudo code for the staggered inflation (Procedure \stagInflate) that guarantees $O(\log n)$ rounds and messages (whp) in the \emph{worst case}.}
\label{algo:inflate}
\end{algorithm}

\begin{algorithm}[h!]
  \small
\begin{algorithmic}[1]
\item[] \textbf{Given:} current network size $n$ (as computed
  by $\computeLow$). All virtual vertices and all nodes are unmarked.
\item[]
\item[] \textbf{\boldmath Phase 1. Compute smaller $p$-cycle:}
\STATE Node $u$ forwards a deflation request through the entire network. 
\STATE Initiating node $u$ floods a request to all other nodes to run this
procedure simultaneously; takes $O(\log n)$ time. 
\STATE Since every node $u$ knows the same virtual graph $\Z_{t-1}(p_i)$ of size
$p_i$, all nodes
locally compute the same prime $p_{s} \in (p_i/8,p_i/4)$ and therefore the
same virtual expander $\Z_t(p_s)$ with vertex set $\mathbb{Z}_{p_{s}}$.
\STATE (Compute the new set of locally simulated virtual vertices
$\NewSim(u) \subset \Z_t(p_s)$.) \\
Let $\alpha = \frac{p_{i}}{p_{s}}$.
For every $x \in \Sim(u)$ (i.e.\ $x\in \Z_{t-1}(p_i)$) we compute 
$y_x = \lfloor\frac{x}{\alpha} \rfloor$. \\
If there is no $x'<x$ such that $y_{x'} = y_x$, we add $y_x$ to $\NewSim(u)$.
This yields the (possibly empty) set \\
$\NewSim(u)=\{y_{x_1},\dots,y_{x_k}\}$, \\ 
where $x_1,\dots,x_k \in \Z_{t-1}(p_{i})$ are
a subset of the previously simulated vertices at $u$.
If $\NewSim(u)=\emptyset$, we mark $u$ as \emph{contending}.
For every vertex $y_{x_j}$, we set\\
$\cloud{y_{x_j}} = \{ m \colon  (m-1)\lfloor\alpha\rfloor \le y_{x_j}<m\lfloor\alpha\rfloor \}$.
\FOR{every $y_{x_j} \in \NewSim(u)$, $(1\le j\le k)$, } 
\item[] (Compute the new set of edges.)\\
\begin{compactdesc}
\item Cycle edges: Add an edge between $u$ and the nodes $v$ and $v'$
  that simulate
  $y_{x_j}-1$ and $y_{x_j}+1$ by using the cycle edges of $\Z_{t-1}(p_i)$ in $G_t$.
\item Inverse edges: Add an edge between $u$ and the node $v$ that
  simulates $y_k^{-1}$; node
  $v$ is found by solving a permutation routing instance.
\end{compactdesc}
\ENDFOR
\item[]
\item[] \textbf{Phase 2. Ensure virtual mapping:}
\IF{$\Sim(v)=\emptyset$}
\STATE Node $v$ marks itself as \emph{contending}.
\ELSE
\STATE Node $v$ reserves one vertex $z\in \Sim(v)$ for itself by marking $z$ as \emph{taken}.
\ENDIF
\WHILE{$v$ is \emph{contending}}
    \STATE Every \emph{contending} node $v$ performs a random walk of
    length $T=\Theta(\log n)$ on the virtual graph $\Z_t(p_{i+1})$ by
    forwarding a token $\tau_v$. This walk is simulated on the actual
    network $U_t$ (with constant overhead). 
    To account for congestion, we give this walk $\rho=O(\log^2 n)$
    rounds to complete; after $T$ random steps, the token remains at its
    current vertex.
    \STATE If, after $\rho$ rounds, $\tau_v$ has reached a virtual
    vertex $z$ (simulated at some node $w$), no other
    token is currently at $z$, and $z$ is not marked as \emph{taken}, then
    $v$ marks itself as \emph{non-contending} and requests $z$ to be
    transfered from $w$ to $v$ where it is marked as \emph{taken}.
    \ENDWHILE
\end{algorithmic}
\caption{\small Procedure \deflate. This is a simplified deflation procedure yielding amortized bounds. Note that Procedure \stagDeflate\ provides the same functionality using  $O(\log n)$ rounds and message whp even in the \textbf{worst case}.}
\label{algo:deflate}
\end{algorithm}
}

\onlyShort{\noindent\textbf{Worst Case Bounds for Type-2 Recovery:}}
\onlyLong{\subsection{Worst Case Bounds for Type-2 Recovery} \label{sec:improved}}
Whereas Lemma~\ref{lem:atMostSimple} shows $O(\log n)$ worst case bounds 
for steps with type-1 recovery,
handling of type-2 recovery that we have 
described so far yields \emph{amortized} polylogarithmic performance 
guarantees on messages and rounds w.h.p.\ per
step\onlyLong{ (cf.\ Cor.~\ref{cor:runtime})}.
We now present a  more complex algorithm for
type-2 recovery that yields worst case
logarithmic bounds on messages and rounds  per step (w.h.p.). 
The main idea of Procedures $\lightInflate$ and $\lightDeflate$\onlyShort{ (cf.\ full paper \cite{full})} is to spread 
the type-2 recovery over $\Theta(n)$ steps of type-1 recovery,
while still retaining constant node degrees and spectral expansion in \emph{every} step.
\onlyShort{The details of these procedures are described in the full paper \cite{full}.

These operations are orchestrated by a \emph{coordinator} node, which is the node that currently hosts
the virtual vertex with integer-label $0$; the coordinator keeps
track of additional information (requiring $O(\log n)$ memory): the current network size $n$ and the sizes of $\Low$ and $\High$ (but not the actual network topology), as follows:

Recall that we start out with an initial network of constant size, thus
initially coordinator $w$ can compute these values with constant overhead.
If an insertion or deletion of some neighbor of $v$ occurs and the algorithm
performs type-1 recovery,
then $v$ informs coordinator $w$ of the changes to the network
size and the sizes of $\High$ and $\Low$ (by routing a message along a
shortest path in $\Z_{t-1}(p_i)$) at the end of the type-1 recovery.
Node $v$ itself simulates some vertex $x \in \mathbb{Z}_{p_i}$ and hence 
can locally compute a shortest path from $x$ to $0$ (simulated at $w$) 
according to the edges in $\Z_t(p_i)$ (cf.\ Fact~\ref{fact:distances}).
The neighbors of $w$ replicate $w$'s state and update their copy in every step.
If the coordinator $w$ itself is deleted, the neighbors transfer its state to
the new coordinator that subsequently simulates $0$.
The coordinator state requires only $O(\log n)$ bits and thus can be sent in $1$
message.
Keep in mind that the coordinator does \emph{not} keep track of the actual
network topology or $\High$ and $\Low$, as this would require $\Omega(n)$
rounds for transferring the state to a new coordinator.

The following lemma summarizes the worst case bounds for staggered inflation/deflation via a coordinator node (cf.\ full paper \cite{full} for the proof):
\begin{lemma}[Worst Case Bounds Type-2 Rec.] \label{lem:advanced}
  Suppose that nodes initiate either $\stagInflate$ or $\stagDeflate$ during recovery in a
  step $t_0$ and $G_{t_0-1}$ is $4\zeta$-balanced.
  Then, for all $t \in [t_0,t_0+T]$ where $T=\lceil 2\theta n \rceil$ the
  following hold:
  \begin{compactitem}
    \item[(a)] Every node simulates at most $8\zeta$ vertices, and the recovery in
  $t$ requires at most $O(\log n)$ rounds and messages (w.h.p.), while making only
    $O(1)$ changes to the topology.
  \item[(b)] The spectral gap of $G_t$ is at least $\frac{(1-\lambda)^2}{8}$
   where $1-\lambda$ is the spectral gap of the $p$-cycle expander family.
 \end{compactitem}
\end{lemma}
Note that staggering the inflation/deflation yields a slightly worse (but nevertheless constant) spectral gap, as stated in Part (b) of Lemma~\ref{lem:advanced}.
This phenomenon is caused by the additional edges (of the new $p$-cycle) that are added on top of the current (old) $p$-cycle.
Since the edge expansion of our expander construction is lower bounded by the edge expansion of the old $p$-cycle, we can still get a constant bound on the spectral gap by invoking the Cheeger inequality (cf.\ Theorem 2.6 in \cite{Wigderson-exsurvey}).
}

\onlyLong{\paragraph{The coordinator}
\onlyShort{\noindent\textbf{The coordinator:}}
The node $w$ that currently simulates the virtual vertex with integer-label $0 \in
V(\Z_{t-1}(p_i)) = \mathbb{Z}_{p_i}$ is called \emph{coordinator} and keeps
track of the current network size $n$ and the sizes of $\Low$ and $\High$ 
as follows:
Recall that we start out with an initial network of constant size, thus
initially coordinator $w$ can compute these values with constant overhead.
If an insertion or deletion of some neighbor of $v$ occurs and the algorithm
performs type-1 recovery,
then $v$ informs coordinator $w$ of the changes to the network
size and the sizes of $\High$ and $\Low$ (by routing a message along a
shortest path in $\Z_{t-1}(p_i)$) at the end of the type-1 recovery.
Node $v$ itself simulates some vertex $x \in \mathbb{Z}_{p_i}$ and hence 
can locally compute a shortest path from $x$ to $0$ (simulated at $w$) 
according to the edges in $\Z_t(p_i)$ (cf.\ Fact~\ref{fact:distances}).
The neighbors of $w$ replicate $w$'s state and update their copy in every step.
If the coordinator $w$ itself is deleted, the neighbors transfer its state to
the new coordinator that subsequently simulates $0$.
The coordinator state requires only $O(\log n)$ bits and thus can be sent in $1$
message.
Keep in mind that the coordinator does \emph{not} keep track of the actual
network topology or $\High$ and $\Low$, as this would require $\Omega(n)$
rounds for transferring the state to a new coordinator.

\onlyShort{\noindent\textbf{Staggering the Inflation.}}
\onlyLong{\subsubsection{Staggering the Inflation} \label{sec:stagInflate}}
We proceed in 2 phases each of which is staggered over $\lceil\theta
n\rceil$ steps.
Let $PC$ denote the $p$-cycle at the beginning of the inflation step.
If, in some step $t_0$ the coordinator is notified (or notices itself) that $|\High| < 3\theta n$, it initiates (staggered)
inflation to build the new $p$-cycle $PC'$ on $\mathbb{Z}_{p_{i+1}}$ by sending a request to
the set of nodes $I$ that simulate the set of vertices $S=\{1,\dots,\lceil
  1/\theta \rceil\}$.
The $\lceil 1/\theta \rceil$ nodes in $I$ are called \emph{active in step $t_0$}. 

\onlyLong{\paragraph{Phase~1: Adding a larger $p$-cycle}}
\onlyShort{\noindent\textbf{Phase~1: Adding a larger $p$-cycle.}}
For every $x \in S$, the simulating node in $I$ adds a cloud of 
vertices as described in Phase~1 of $\inflate$.
More specifically, for vertex $x$ we add a set $Y \subset V(PC')$ of $c(x)$ vertices, as
defined in Eq.~\eqref{eq:cx} on page~\pageref{eq:cx}.
We denote this set of new vertices by $\NewSim(v)$.
That is, node $v$ now simulates $|\Load(v)| + |\NewSim(v)|$ many vertices.
In contrast to $\inflate$, however, vertex $x\in PC$ and its edges
are \emph{not} replaced by $Y$ (yet).
For each node in $y \in Y$, the simulating node $v$ computes the cycle edges and
inverse $y^{-1}\in PC'$.
It is possible that $y^{-1}$ is not among the vertices in $S$, and hence is
not yet simulated at any node in $I$.
Nevertheless, by Eq.~\eqref{eq:cx}, $v$ can locally compute the vertex $x'\in
PC$ that is going to be inflated to the cloud that contains $y^{-1} \in PC'$.
Therefore, we add an \emph{intermediate edge} $(y,x')$, which requires $O(\log n)$
messages and rounds.
Note that $|\NewSim(v)|$ could be as large as $4\zeta^2$.
Therefore, similarly as in Phase~2 of $\inflate$, a node in $I$ needs to redistribute newly generated vertices if $|\NewSim| > 4\zeta$ as follows:
The nodes in $I$ proceed by performing random walks to find node with 
small enough $\NewSim$.
Note that, even though $\stagInflate$ has not yet been processed at nodes in
$V(G_t) \setminus I$, any node that is hit by this random walk can
locally compute its set $\NewSim$ and thus check if it is able to simulate an additional vertex in the next $p$-cycle $PC'$.
Since we have $O(1)$ nodes in $I$ each having $O(1)$ vertices in their $\NewSim$ set, these walks can be done \emph{sequentially}, i.e., only $1$ walk is in 
progress at any time, which takes $O(\log n)$ rounds in total.

After these walks are complete and all nodes in $I$ have $|\NewSim|\le
4\zeta$, the coordinator is notified and forwards the inflation request to nodes $I'$ that simulate vertices $S'=\{\lceil 1/\theta \rceil + 1,\dots,2\lceil 1/\theta
\rceil\}$. (Again, this is done by locally computing the shortest path in
$PC$.)
In step $t_0+1$, the nodes in $I'$ become \emph{active} and proceed the same way
as nodes in $I$ in step $t_0$, i.e., clouds and intermediate edges are added for
every vertex in $S'$.


\onlyLong{\paragraph{Phase~2: Discarding the old $p$-cycle.}}
\onlyShort{\noindent\textbf{Phase~2: Discarding the old $p$-cycle.}}
Once Phase~1 is complete, i.e., all nodes are simulating the vertices in their respective $\NewSim$ set, the coordinator sends another request to the set of nodes $I$ ---the \emph{active nodes} in the next step---that are still simulating the
set $S$ of the first $\lceil 1/\theta \rceil$ vertices in the old $p$-cycle $PC$.
Every node in $I$ drops all edges of $PC$ and stops simulating vertices in $V(PC)$.
In the next step, this request is forwarded to the nodes that simulate the next $\lceil \theta n\rceil$ vertices and reaches all nodes within $\theta n$ steps.
After $T= 2\theta n$ steps\footnote{For clarity of presentation, we assume that $2\theta n$ is an integer.}, the inflation has been processed at all nodes.

Finally, we need to argue that type-1 recovery succeeds with high probability while the staggered inflation is ongoing:
If the adversary inserts a node $w$ in any of these $T$ steps, we can simply assign one of the newly inflated vertices to $w$.
If, on the other hand, the adversary deletes nodes, we need to show that, for any $t \in [t_0,t_0+T]$, it holds that $|\Low_t| \ge \theta n$.
Recalling that the coordinator invoked the inflation in step $t_0$ because $|\Spare_{t_0}| < 3\theta n$, it follows that $|\Low_{t_0}| \ge n - 3\theta n$.
In the worst case, the adversary deletes $1$ node in every one of the following $T$ steps, which increases the load of at most $2\theta n$ nodes.
This yields that $|\Low_t| \ge |\Low_{t_0}| - 2\theta n = n - 5\theta n \ge \theta n$, due to \eqref{eq:theta}.
Thus, since the assumption of Lemma~\ref{lem:gillman}.(b) holds throughout steps $[t_0,t_0+T]$, type-1 recovery succeeds with high probability as required.

\onlyLong{\subsubsection{Staggering the Deflation} \label{sec:stagDeflate}

We now describe the implementation of $\stagDeflate$ that yields a worst case
bound of $O(\log n)$ for the recovery in every step. 
Similarly to $\stagInflate$, the coordinator initiates a staggered deflation
whenever the threshold $|\Low| < 3\theta$ is reached and the algorithm proceeds
in two phases:

\paragraph{Phase 1: Adding a smaller $p$-cycle} 

Phase~1 is initiated during the recovery in some step $t_0$ by the
(current) coordinator $w$ who sends a message to nodes $S$ that simulate
vertices $I=\{1,\dots,\lceil 1/\theta \rceil\}$.
The nodes in $S$ become \emph{active} in the recovery of step $t_0$ and will start
simulating the (smaller) $p$-cycle $\Z(p_s)$ in addition to the current
$p$-cycle $\Z_{t_0}(p_i)$ by the end of the step, as described below.
As in the case of $\stagInflate$, $w$ can efficiently find $S$
(requiring only $O(\log n)$ messages and rounds) by following the shortest path
in the current $p$-cycle $\Z_{t_0}(p_{i})$.
Let $\alpha = p_i / p_s$ and consider some node $v \in S$.
For every $x \in \Sim(v)$, node $v$ computes $y_x = \lfloor x/\alpha
\rfloor$ and starts simulating $y_x \in \Z(p_s)$, if there is no $x' < x$ such
that $x' = \lfloor x' / \alpha \rfloor$.
That is, the new vertices are determined exactly the same way as in
Phase~1 of $\deflate$ and node $v$ adds $y_x$ to $\NewSim(v)$.

Assuming that there is a $y_x \in \NewSim(v)$, node $v$ marks all
$x_1,\dots,x_k \in \Z_{t_0}(p_i)$ that satisfy $y_x = \lfloor x_j / \alpha
\rfloor$, for $1\le j \le k$, as \emph{taken}.
We say that $x$ \emph{dominates} $x_1,\dots,x_k$ and we call the set
$\{x_1,\dots,x_k\}$ a \emph{deflation cloud}.
Note that some of the vertices of a deflation cloud might be simulated at other
nodes.
Nevertheless, according to the edges of $\Z_{t_0}(p_i)$, these nodes are in an
$O(1)$ neighborhood of $v$ and can thus be notified to mark the corresponding
vertices as \emph{taken}.
Intuitively speaking, if a node $v$ simulates such a dominating vertex
$x$, then $v$ is guaranteed to simulate a vertex in the new $p$-cycle
$\Z(p_s)$, and the surjective requirement of the virtual mapping is
satisfied at $v$.
Thus our goal is to ensure that every node in $S$ simulates a dominating vertex
by the end of the recovery of this step.

The problematic case is when none of the vertices currently simulated at node
$v$ dominates for some $y_x \in \Z(p_s)$.
To ensure that $v$ simulates at least $1$ vertex of the new $p$-cycle
$\Z(p_s)$, node $v$ initiates a random walk on the graph $\Z(p_s)$ to find a
dominating vertex that has not been marked \emph{taken}.
We thus lower-bound the size of dominating vertices that are never marked as
\emph{taken}, in any of the $\theta n$ steps during which $\stagDeflate$ is in
progress:

Recall that the coordinator invoked $\stagDeflate$ because $|\Low|<3\theta$.
This means that $\ge (1-3\theta)n$ nodes have $\Load_{t_0} > 2\zeta$ and the total
load in the network is at least $(2\zeta(1-3\theta)+3\theta)n$ since every
node simulates at least $1$ vertex.
If some node simulates a dominating vertex $x$, then \emph{all} of the (at
most $\alpha \le 8$) dominated vertices $x'>x$ that also satisfy $y_x = \lfloor x' /
\alpha \rfloor$ are marked as \emph{taken}.
Considering that $\zeta \le 8$, the number of dominating vertices is at least
$(2\zeta(1-3\theta)+3\theta)n/8  \ge (2 - \theta(6+3/\zeta))n$.
In each of the $\theta n$ steps while Phase~1 of $\stagDeflate$ is in progress, the adversary might insert some node that starts simulating a dominating vertex.
Thus, in total we must give up $n+\theta n$ dominating vertices.
It follows that the number of dominating vertices that are available
(i.e.\ not needed by any node) is at least
$$
(2 - \theta(6+3/\zeta))n - n  - \theta n = (1-\theta(6+3/\zeta+1))n.
$$
Recalling \eqref{eq:theta} on page~\pageref{eq:theta}, the right
hand size is at least a constant fraction of $n$, i.e., the set of available
dominating vertices $D$ has size $\ge \eps n$ while $\stagDeflate$ is in
progress, for some $\eps >0$.
Similarly to the proof of Lemma~\ref{lem:gillman}, we can use the concentration bound of \cite{G1998:Chernoff} to show that a random walk of $v$ of length $O(\log n)$ hits a vertex in $D$ with high probability.
To avoid clashes between nodes in $S$, we perform these walks sequentially.
Since there are only $O(1)$ nodes in $S$, this takes overall $O(\log n)$ time
and messages. 

In step $t_0+1$, the nodes that simulate the next $1/\theta$ vertices become
active and so forth, until the request returns to the (current)
coordinator after $\lceil \theta n\rceil$ steps.

\onlyLong{\paragraph{Phase~2: Discarding the old $p$-cycle}}
\onlyShort{\noindent\textbf{Phase~2: Discarding the old $p$-cycle.}}
Once the new (smaller) $p$-cycle $\Z(p_s)$ has been fully constructed,
the coordinator sends another request to the nodes in $I$---which again
become \emph{active nodes}---that simulate the $\lceil 1/\theta \rceil$
vertices in $S$.
Every node in $I$ drops all edges of $E(\Z(p_{i}))$ and stops simulating
vertices in $V(\Z(p_{i}))$.
This request is again forwarded to the nodes that simulate the next $\theta n$
vertices and finally has reached all nodes within $\theta n$ steps.
Thus, after $T=\lceil 2\theta n \rceil$ steps, the deflation has been
completed at all nodes.

Since the coordinator initiated the deflation because $|\Low_{t_0}| < 3\theta n$, it follows that $|\Spare_{t_0}|\ge n - 3\theta n$, and thus $|\Spare_t|\ge \theta n$, for all steps $t \in [t_0,t_0+T]$.
Therefore, by an argument similar to Procedure $\stagInflate$, it follows that
type-1 recovery succeeds w.h.p.\ until the new virtual graph is in place.

\begin{lemma}[Worst Case Bounds Type-2 Recovery] \label{lem:advanced}
  Suppose that the coordinator initiates either $\stagInflate$ of $\stagDeflate$ during recovery in some
  step $t_0$ and $G_{t_0-1}$ is $4\zeta$-balanced.
  Then, for all steps $t \in [t_0,t_0+T]$ where $T=\lceil 2\theta n \rceil$ the
  following hold:
  \begin{compactenum}
  \item[(a)] Every node simulates at most $8\zeta$ vertices and the recovery in
    $t$ requires at most $O(\log n)$ rounds and messages (w.h.p.), while making only
    $O(1)$ changes to the topology.
 \item[(b)] The spectral gap of $G_t$ is at least $\frac{(1-\lambda)^2}{8}$
   where $1-\lambda$ is the spectral gap of the $p$-cycle expander family.
\end{compactenum}
\end{lemma}
\enlargethispage{\baselineskip}
\begin{proof}
First consider (a):
The bound of $8\zeta$ vertices follows from the fact that, during
$\stagInflate$ and $\stagDeflate$, any node simulates at most $4\zeta$
vertices from both $p$-cycles.
This immediately implies a constant node degree.
Recalling the description of Phases~1 and 2 for $\stagInflate$ and
$\stagDeflate$, we observe that either phase causes an overhead of $O(\log n)$
messages and rounds for each of the $O(1)$ active nodes during recovery in some step $t \in [t_0,t_0+T]$; the worst case bounds of (a) follow.

We now argue that, at any time during the staggered inflation, we
still guarantee a constant spectral gap.
By the left inequality of Theorem~\ref{thm:cheeger} 
(App.~\ref{app-sec:preliminaries}), a spectral
expansion of $\lambda_{G_{t_0-1}}$ yields an edge expansion (cf.\
Def.~\ref{def:edgeexpansion} in App.~\ref{app-sec:preliminaries})
$h(G_{t_0-1}) \ge (1-\lambda_{G_{t_0-1}})/2$, which is $O(1)$.
For both, $\stagInflate$ and $\stagDeflate$, it holds that during Phase~1,
nodes still simulate the full set of vertices and edges of the old $p$-cycle
and some intermediate edges of the new $p$-cycle.
In Phase~2, on the other hand, nodes simulate a full set of vertices and edges
of the new $p$-cycle and some edges of the old $p$-cycle.
Thus, during either phase, the edge expansion is
bounded from below by the edge expansion of the $p$-cycle expander family. 
That is, we have $h(G_{t}) \ge h(G_{t_0-1})$, for any step $t \in [t_0,t_0+T]$.
It is possible, however, that the additional intermediate edges decrease the
spectral expansion.
Nevertheless, we can apply the right inequality of Theorem~\ref{thm:cheeger} to
get $$1-\lambda_{G_{t}} \ge \frac{h^2(G_{t_0-1})}{2} \ge (1-\lambda_{G_{t_0-1}})^2 / 8,$$ as required.
\end{proof}

\subsubsection{Proof of Theorem~\ref{thm:main}}
Lemmas~\ref{lem:atMostSimple} and \ref{lem:advanced} imply the sought worst case
bounds of Theorem~\ref{thm:main}.
The constant node degree follows from Lemma~\ref{lem:atMostSimple}.(a) and
Lemma~\ref{lem:advanced}.(a).
Moreover, Lemma~\ref{lem:advanced}.(b) shows a constant spectral gap for
(the improved) type-2 recovery steps and the analogous result for type-1
recovery follows from Lemma~\ref{lem:fanChung} and
Lemma~\ref{lem:atMostSimple}.(a).
}
%
}
\onlyShort{\noindent{\bf Implementing a Distributed Hash Table (DHT):} \label{sec:dht}}
\onlyLong{\subsubsection{\bf Implementing a Distributed Hash Table (DHT)} \label{sec:dht}}
We can leverage our expander maintenance algorithm to implement a DHT as follows:
Recall that the current size $s$ of the $p$-cycle is global knowledge.
Thus every node uses the same hash function $h_s$, which uniformly maps keys to the vertex set of the $p$-cycle.

\onlyShort{ We consider the case where no staggered inflation/deflation is in progress and defer the more complex case to the full paper \cite{full}:}%
\onlyLong{We first look at the case where no staggered inflation/deflation is in progress:}
If some node $u$ wants to store a key value pair $(k,val)$ in the DHT, $u$
computes the index $z:=h_s(k)$.
Recall that $u$ can locally compute a shortest path $z_1,z_2,\dots,z$ (in the $p$-cycle) starting at one of its simulated virtual vertices $z_1$ and ending at vertex $z$.
Even though node $u$ does not know how this entire path is mapped to the actual network, it can locally route by simply forwarding $(k,val)$ to the neighboring node $v_2$ that simulates $z_2$; node $v_1$ in turn forwards the key value pair to the node that simulates $z_3$ and so forth.
The node that simulates vertex $z$ stores the entry $(k,val)$.
If $z$ is transferred to some other node $w$ at some point, then storing $(k,val)$ becomes the responsibility of $w$.
Similarly, for finding the value associated with a given key $k'$, node $u$ routes a message to the node simulating vertex $h_s(k')$, who returns the associated value to $u$.
It is easy to see that insertion and lookup both take $O(\log n)$ time and $O(\log n)$ messages and that the load at each node is balanced.
\onlyLong{

We now consider the case where a staggered inflation (cf.\ Procedure~\ref{algo:stagInflate}) has been started and some set of nodes have already constructed the next larger $p$-cycle of size $s'$.
Let $PC$ be the old (but not yet discarded) $p$-cycle and let $PC'$ denote the new $p$-cycle that is currently under construction.
For a given vertex $z_i \in PC$ we use the notation $z_i'$ to identify the unique vertex in $PC'$  that has the same integer label as $z_i$.

Note that all nodes have knowledge of the hash function $h_{s'}$, which maps to the vertices of $PC'$.
Suppose that a node $u \in S $ becomes active during Phase~1 of the staggered inflation and starts simulating vertices $z_1',\dots,z_\ell' \in PC'$.
(For clarity of presentation, we assume that $\ell \le 4\zeta$, thus $u$ does not need to redistribute these vertices. The case where $\ell > 4\zeta$ can be handled by splitting the operations described below among the nodes that end up simulating $z_1',\dots,z_\ell'$.)
At this point, some set $S$ of $j$ nodes might still be simulating the corresponding vertices $z_1,\dots,z_\ell \in PC$, where $j\le \ell \in O(1)$.
Thus node $u$ contacts the nodes in $S$ (by routing a message to vertices $z_1,\dots,z_\ell$ along the edges of $PC$) and causes these nodes to transfer all data items associated with $z_1,\dots,z_\ell$ to $u$.
From this point on until the staggered inflation is complete, the nodes in $S$ forward all insertion and lookup requests regarding $z_1,\dots,z_\ell$ to node $u$.
Note that the above operations require at most $O(\log n)$ rounds and messages, and thus only increase the complexity of the staggered inflation by a constant factor.

The case where a staggered deflation is in progress is handled similarly, by transferring key value pairs of vertices that are contracted to a single vertex in the new (smaller) $p$-cycle, whenever the simulating node becomes active.
}


%% file: Conclusions.tex
\section{Conclusion}
\vspace{-0.35cm}
\enlargethispage{\baselineskip}
We have presented a distributed algorithm for maintaining an expander
efficiently using only $O(\log n)$ messages and rounds
in the worst case and guarantee a constant spectral gap and node degrees deterministically at all times.
There are some open questions: 
Is an $O(\log n)$ overhead sufficient for handling even a linear
number of insertion/deletions per step? 
How can we deal with malicious nodes in this setting?

%% file: AlgorithmMain.tex
\label{sec:pseudo}

\begin{algorithm*}[h!]
\begin{algorithmic}
\item[] 
\item[] \textbf{Case 1: Adversary inserts a node $u$}:
\STATE Try to find a spare vertex for $u$ via a random walk
(\textbf{type-1 recovery}). 
\IF{type-1 recovery fails} 
  \IF{most nodes simulate only $1$ vertex} 
    \STATE Perform \textbf{type-2 recovery} by inflating.
  \ELSE
    \STATE Retry type-1 recovery until it succeeds.
  \ENDIF
\ENDIF
\item[] 
\item[] \textbf{Case 2: Adversary deletes a node $u$}:
\STATE Try distributing vertices that were simulated at $u$ via random walks
(\textbf{type-1
recovery}). 
\IF{type-1 recovery fails} 
  \IF{most nodes simulate many vertices} 
    \STATE Perform \textbf{type-2 recovery} by deflating.
  \ELSE
    \STATE Retry type-1 recovery until it succeeds.
  \ENDIF
\ENDIF
\end{algorithmic}
\caption{High-level overview of our algorithm}
\label{algo:highlevel}
\end{algorithm*}

\begin{algorithm}[h!]
\begin{algorithmic}[1]
\item[]{\textbf{Assumption:} the adversary attaches inserted node $u$ to
    arbitrary node $v$}
\item[] 
\item[] \COMMENT{Try to perform a \textbf{type-1 recovery}:}
\STATE Node $v$ initiates a random walk of length $\ell\log n$ by
generating a token $\tau$ and sending it to a neighbor $u'$ chosen
uniformly at random, but excluding $u$.
Node $u'$ in turn forwards $\tau$ by chosing a neighbor at random and
so forth.
Note that the newly inserted node $u$ is excluded from being reached by
the random walk.
The walk terminates upon reaching a node $w \in \High$ (cf.\
Equation~\eqref{eq:spare}).
\label{line:sampleHighDegree}

\IF{found node $w \in \High$}
  \STATE Transfer a virtual vertex and all its edges (according to
  the virtual graph) from $w$ to $u$. Remove edge between $u$ and $v$ unless
  required by $\Z_t$.
\item[]
\ELSE[the walk did not hit a node in $\High$; perform \textbf{type-2 recovery} if necessary:]
  \STATE Determine current network size $n$ and $|\High|$ via
  $\computeHigh$ (cf.\ Algorithm~\ref{algo:computeHighLow}).
  \IF[Perform type-2 recovery:]{$|\High| < \theta n$}
    \STATE Invoke $\inflate$ (cf.\ Algorithm~\ref{algo:inflate}).
  \ELSE[Sufficiently many nodes with spare virtual vertices are present but
  the walk did not find them. Happens with probability $\le 1/n$.]
    \STATE Repeat from Line~\ref{line:sampleHighDegree}.
  \ENDIF 
\ENDIF
\end{algorithmic}
\caption{\texttt{insertion($u, \theta$)}}
\label{app-algo:insertion}
\end{algorithm}

\begin{algorithm}[h!]
\begin{algorithmic}[1]
\item[] \textbf{Assumption:} adversary deletes an arbitrary node $u$ which
  simulated $k$ virtual vertices. (We prove that $k \in O(1)$).
\STATE A (former) neighbor $v$ of node $u$ attaches all edges of $u$ to itself.

\item[] 
\item[] \COMMENT{Try to perform a \textbf{type-1 recovery}:}
\FOR{each of the $k$ vertices} 
  \STATE Node $v$ initiates a random walk of length $\ell\log n$ by
  generating a token $\tau$ and sending it to a uniformly at random chosen
  neighbor $u'$.
  Node $u'$ in turn forwards $\tau$ by chosing a neighbor at random and
  so forth.
  The walk terminates upon reaching a node $w \in \Low$ (cf.\
  Equation~\eqref{eq:low}).
\ENDFOR
\label{line:ifSampleLowDegree}
\IF{all random walks found nodes $w_1,\dots,w_k \in \Low$:}
  \STATE Distribute the virtual vertices of $u$ and their respective edges (according to
  the virtual graph) from $v$ to $w_1,\dots,w_k$. 
\item[]
\ELSE[Some of the random walks did not find a node in $\Low$; perform
\textbf{type-2 recovery} if necessary:]
  \STATE Determine network size $n$ and $|\Low|$ via $\computeLow$ (cf.\
  Algorithm~\ref{algo:computeHighLow}).
  \IF[Perform type-2 recovery:]{$|\Low| < \theta n$}
    \STATE Invoke $\deflate$ (cf.\ Algorithm~\ref{algo:deflate}).
  \ELSE[Sufficiently many nodes with low load are present but the walk(s) did not find
  them. This happens with probability $\le 1/n$:]
    \STATE  Repeat from Line~\ref{line:ifSampleLowDegree}.
  \ENDIF 
\ENDIF
\end{algorithmic}
\caption{Procedure \texttt{deletion($u,\theta$)}}
\label{app-algo:deletion}
\end{algorithm}

\begin{algorithm}[h!]
\begin{algorithmic}[1]
\item[] \textbf{Given:} $\textsc{diam}$ is the diameter of  $\Z_t$ (i.e. $\textsc{diam} \in O(\log n)$).
\STATE Node $u$ broadcasts an aggregation request to all its neighbors. In
addition to the network size, this request indicates whether to compute
$|\Low|$ or $|\High|$. That is, the request of $u$ traverses the network in
a BFS-like manner and then returns the aggregated values to $u$.
\STATE If a node $w$ receives this request from some neighbor, it computes
the aggregated maximum value, according to whether $w \in \High$ for $\computeHigh$ (resp.\
$w\in\Low$ for $\computeLow$).
\STATE If node $w$ has received the request for the first time, $w$
forwards it to all neighbors (except $v$).
\STATE Once the entire network has been explored this way, i.e., the
request has been forwarded for $\textsc{diam}$ rounds, the aggregated
maximum values of the network size and $|\Low|$ (resp.\ $|\High|$) are
sent back to $u$, which receives them after $\le 2\textsc{diam}$
rounds.  
\end{algorithmic}
\caption{Procedures $\computeHigh$ and $\computeLow$.}
\label{algo:computeHighLow}
\end{algorithm}

\begin{algorithm}[h!]
\begin{algorithmic}[1]
\item[] \textbf{Given:} current network size $n$ (as computed
  by $\computeHigh$). All virtual vertices and all nodes are unmarked.
\item[]
\item[] \textbf{\boldmath Phase 1. Compute larger $p$-cycle:}
\STATE Inserted node $u$ forwards an inflation request through the entire network. 
\STATE Initiating node $u$ floods a request to all other nodes to run this
process simultaneously; takes $O(\log n)$ time. 
\STATE Since every node $u$ knows the same virtual graph $\Z_{t-1}(p_i)$, all nodes
locally compute the same prime $p_{i+1} \in (4p_i,8p_i)$ and therefore the
same virtual expander $\Z_t(p_{i+1})$ with vertex set $\mathbb{Z}_{p_{i+1}}$.
\STATE (Compute the new set of locally simulated virtual vertices.) \\
Let $\alpha = \frac{p_{i+1}}{p_i}$ and define the function 
\begin{equation} \label{eq:cxInflate}
c(x) = \lfloor \alpha(x+1)\rfloor - \lfloor \alpha x\rfloor - 1.
\end{equation}
Replace every $x \in \Sim(u)$ (i.e.\ $x\in \Z_{t-1}(p_{i})$) with a cloud of virtual vertices $y_0,\dots,y_{c(x)}$ where 
$y_k = (\lfloor \alpha x \rfloor + k) \mod\ p_{i+1}$, for $0 \le k \le c(x)$.
That is, $\cloud{y_0}=\cdots=\cloud{y_{c(x)}}=\{y_0,\dots,y_{c(x)}\}$.

\FOR{every $x \in \Sim(u)$ and every $y_k$, $(0\le k\le c(x))$} 
\item[] (Compute the new set of edges.)\\ 
\begin{compactdesc}
\item Cycle edges: Add an edge between $u$ and the nodes $v$ and $v'$
  that simulate
  $y_{k}-1$ and $y_{k}+1$ by using the cycle edges of $\Z_{t-1}(p_i)$ in $G_t$.
\item Inverse edges: Add an edge between $u$ and the
node $v$ that simulates $y_k^{-1}$; node $v$ is found by 
solving a permutation routing instance.
\end{compactdesc}
\ENDFOR

\STATE After the construction of $\Z_t(p_{i+1})$ is complete, we
transfer a (newly generated) virtual vertex to the inserted node $u$ from
its neighbor $v$.
\item[]
\item[] \textbf{Phase 2. Perform load balancing:}
\IF{a node $w$ has $\Load(w)>2\zeta$ (i.e.\ $w\notin\Low$)}
\STATE Node $w$ marks all vertices in $\Sim(w)$ as \emph{full}.
\ENDIF
\IF{a node $v$ has load $k' > 4\zeta$ vertices}
  \item[] (Distribute all except $4\zeta$ vertices to other nodes.) 
  \FOR{each of the $k'-4\zeta$ vertices}
  \STATE Node $v$ marks itself as \emph{contending}.
  \WHILE{$v$ is \emph{contending}}
    \STATE Every \emph{contending} node $v$ performs a random walk of
    length $T=\Theta(\log n)$ on the virtual graph $\Z_t(p_{i+1})$ by
    forwarding a token $\tau_v$.
    This walk is simulated on the actual network $U_t$ (with constant
    overhead).
    To account for congestion, we give this walk $\rho=O(\log^2 n)$
    rounds to complete; once a token has taken $T$ steps it remains at its
    current vertex.
    \STATE If, after $\rho$ rounds, $\tau_v$ has reached a virtual
    vertex $z$ (simulated at some node $w$), no other
    token is currently at $z$, and $z$ is not marked as $\emph{full}$, 
    then $v$ marks itself as
    \emph{non-contending} and transfers a virtual vertex to $w$.
    Moreover, if the new load of $w$ is $>2\zeta$, we mark all vertices at
    $w$ as \emph{full}. \label{line:inflateTransfer}
  \ENDWHILE
  \ENDFOR
\ENDIF
\end{algorithmic}
\caption{Procedure \inflate. This is a simplified inflation procedure yielding amortized bounds. Note that Procedure \stagInflate\ provides the same functionality using  $O(\log n)$ rounds and message whp even in the \textbf{worst case}.}
\label{algo:inflate}
\end{algorithm}

\begin{algorithm*}[h!]
\begin{algorithmic}[1]
\item[] \textbf{Given:} current network size $n$ (as computed
  by $\computeLow$). All virtual vertices and all nodes are unmarked.
\item[]
\item[] \textbf{\boldmath Phase 1. Compute smaller $p$-cycle:}
\STATE Node $u$ forwards a deflation request through the entire network. 
\STATE Initiating node $u$ floods a request to all other nodes to run this
procedure simultaneously; takes $O(\log n)$ time. 
\STATE Since every node $u$ knows the same virtual graph $\Z_{t-1}(p_i)$ of size
$p_i$, all nodes
locally compute the same prime $p_{s} \in (p_i/8,p_i/4)$ and therefore the
same virtual expander $\Z_t(p_s)$ with vertex set $\mathbb{Z}_{p_{s}}$.
\STATE (Compute the new set of locally simulated virtual vertices
$\NewSim(u) \subset \Z_t(p_s)$.) \\
Let $\alpha = \frac{p_{i}}{p_{s}}$.
For every $x \in \Sim(u)$ (i.e.\ $x\in \Z_{t-1}(p_i)$) we compute 
$y_x = \lfloor\frac{x}{\alpha} \rfloor$. \\
If there is no $x'<x$ such that $y_{x'} = y_x$, we add $y_x$ to $\NewSim(u)$.
This yields the (possibly empty) set \\
$\NewSim(u)=\{y_{x_1},\dots,y_{x_k}\}$, \\ 
where $x_1,\dots,x_k \in \Z_{t-1}(p_{i})$ are
a subset of the previously simulated vertices at $u$.
If $\NewSim(u)=\emptyset$, we mark $u$ as \emph{contending}.
For every vertex $y_{x_j}$, we set\\
$\cloud{y_{x_j}} = \{ m \colon  (m-1)\lfloor\alpha\rfloor \le y_{x_j}<m\lfloor\alpha\rfloor \}$.
\FOR{every $y_{x_j} \in \NewSim(u)$, $(1\le j\le k)$, } 
\item[] (Compute the new set of edges.)\\
\begin{compactdesc}
\item Cycle edges: Add an edge between $u$ and the nodes $v$ and $v'$
  that simulate
  $y_{x_j}-1$ and $y_{x_j}+1$ by using the cycle edges of $\Z_{t-1}(p_i)$ in $G_t$.
\item Inverse edges: Add an edge between $u$ and the node $v$ that
  simulates $y_k^{-1}$; node
  $v$ is found by solving a permutation routing instance.
\end{compactdesc}
\ENDFOR
\item[]
\item[] \textbf{Phase 2. Ensure Surjective Mapping:}
\IF{$\Sim(v)=\emptyset$}
\STATE Node $v$ marks itself as \emph{contending}.
\ELSE
\STATE Node $v$ reserves one vertex $z\in \Sim(v)$ for itself by marking $z$ as \emph{taken}.
\ENDIF
\WHILE{$v$ is \emph{contending}}
    \STATE Every \emph{contending} node $v$ performs a random walk of
    length $T=\Theta(\log n)$ on the virtual graph $\Z_t(p_{i+1})$ by
    forwarding a token $\tau_v$. This walk is simulated on the actual
    network $U_t$ (with constant overhead). 
    To account for congestion, we give this walk $\rho=O(\log^2 n)$
    rounds to complete; after $T$ random steps, the token remains at its
    current vertex.
    \STATE If, after $\rho$ rounds, $\tau_v$ has reached a virtual
    vertex $z$ (simulated at some node $w$), no other
    token is currently at $z$, and $z$ is not marked as \emph{taken}, then
    $v$ marks itself as \emph{non-contending} and requests $z$ to be
    transfered from $w$ to $v$ where it is marked as \emph{taken}.
    \ENDWHILE
\end{algorithmic}
\caption{Procedure \deflate. This is a simplified deflation procedure yielding amortized bounds. Note that Procedure \stagDeflate\ provides the same functionality using  $O(\log n)$ rounds and message whp even in the \textbf{worst case}.}
\label{algo:deflate}
\end{algorithm*}

\begin{algorithm*}
\begin{algorithmic}[1]
  \item[] \textbf{Assumption: } Let node $w$ be the node that simulates vertex $0$.
  \STATE Coordinator $w$ maintains local counters of $|\Spare|$, $|\Low|$
  and the network size $n$.
  \STATE The neighbors of $w$ replicate the state of $w$, i.e., everytime
  $w$ updates any of its counters, it sends a message to all of its
  neighbors. If $w$ itself is deleted, normal recovery is performed to
  find a node $w'$ to take over vertex $0$.
  Then, the neighbors transfer the coordinator state to the new
  coordinator $w'$.
  Recall that, according to the virtual graph structure, all former
  neighbors of $w$ become neighbors of $w'$.
  \item[]
  \item[] \textbf{Upon insertion of some node $u$ attached to $v$:}
  \STATE Node $v$ tries to perform type-1 recovery (as in
  $\insertion(u,\theta)$).
  \IF{the recovery succeeds} 
    \STATE Some vertex was transferred to $u$ from some node $u'$.
  Node $v$ sends a message along a shortest
  path in the virtual graph $\Z_{t}$ to the coordinator $w$.
  This message also contains information about changes in the number of
  nodes in $\Spare$ and $\Low$.
  This information only depends on the load at $u'$ and thus does not
  require any additional communication.
    \STATE Coordinator $w$ increases/decreases its local counters
  accordingly. 
  \ELSE
    \STATE Node $v$ sends a request to the coordinator, informing about
    the failed type-1 recovery.
    Coordinator $w$ checks its  (updated) local counters and, if $|\Spare|<3\theta$,
    starts invoking $\stagInflate$.
  \ENDIF
  \item[]
  \item[] \textbf{Upon deletion of some node $u$ previously attached to $v$:}
  \STATE Node $v$ tries to perform type-1 recovery (as in
  $\deletion(u,\theta)$).
  \IF{the recovery succeeds} 
    \STATE The vertices simulated at $u$ were transferred to other nodes
    $u_1',\dots,u_k'$.
    Node $v$ sends a message along a shortest path in $\Z_t$ to the
    coordinator $w$.
This shortest path can be computed locally, since every
node knows the complete virtual graph.
    This message also contains information about changes in the number of
    nodes in $\Spare$ and $\Low$.
    This information only depends on the load at $u_1',\dots,u_k'$ and
    thus does not require additional communication.
    \STATE Coordinator $w$ increases/decreases its local counters
  accordingly. 
  \ELSE
    \STATE Node $v$ sends a request to the coordinator, informing about
    the failed type-1 recovery.
    Coordinator $w$ checks its (updated) local counters and, if $|\Low|<3\theta$,
    starts invoking $\stagDeflate$.
  \ENDIF
\end{algorithmic}
\caption{Advanced handling of type-2 recovery via a coordinator node $w$
  which yields $O(\log n)$ worst case bounds on messages and rounds per
  insertion/deletion.
  (Needed for $\stagInflate$ and $\stagDeflate$.)}
\label{algo:coordinator}
\end{algorithm*}

\begin{algorithm*}[h!]
\begin{algorithmic}[1]
\item[] 
\STATE \textbf{Assumption:} Let $w$ be the \emph{coordinator} node that
maintains local counters of $\Spare$, $\Low$ and the network size (cf.\
Algorithm~\ref{algo:coordinator}). Moreover, the coordinator has computed
the prime number $p_{i+1}$ of the larger $p$-cycle to which we inflate.

\item[]
\item[] \textbf{\boldmath Phase 1. Adding a larger $p$-cycle:}
\STATE The coordinator sends an initiation request to the nodes $I$ that
simulate the vertices $S=\{1,\dots,1/\theta\}$.
This set $I$ are the active nodes in the recovery of the current step.
\item[] 
\item[] (Compute the new set of locally simulated virtual vertices.)
\item[] Every node $u \in I$ does the following:
Let $\alpha = \frac{p_{i+1}}{p_i}$ and define the function 
$c(x) = \lfloor \alpha(x+1)\rfloor - \lfloor \alpha x\rfloor - 1.$
\STATE For every $x \in \Sim(u)$ (i.e.\ $x\in \Z_{t-1}(p_{i})$), node $u$ adds a
cloud of virtual vertices $y_0,\dots,y_{c(x)}$ where $y_k = (\lfloor
\alpha x \rfloor + k) \mod\ p_{i+1}$, for $0 \le k \le c(x)$.  That is,
$\cloud{y_0}=\cdots=\cloud{y_{c(x)}}=\{y_0,\dots,y_{c(x)}\}$.
\STATE Node $u$ adds all such generated vertices $y_i$ to the set
$\NewLoad(u)$.

\FOR{every $x \in \Sim(u)$ and every $y_k$, $(0\le k\le c(x))$} 
\item[] (Compute the new set of edges.)\\ 
\begin{compactdesc}
\item Cycle edges: Add an edge between $u$ and the nodes $v$ and $v'$
  that simulate
  $y_{k}-1$ and $y_{k}+1$ by using the cycle edges of $\Z_{t-1}(p_i)$ in $G_t$.
  In case that $v$ (or $v'$) have not yet been active in Phase~1, we place
  an \emph{intermediate edge} from $u$ to $v$, resp.\ $v'$.
\item Inverse edges: Add an edge between $u$ and the
node that is going to simulate $y_k^{-1}$.
Node $u$ can locally compute the vertex $x'$ (simulated at some node $v'$),
for which the corresponding cloud (containing $y_k^{-1}$) is going to be
added, and hence can add an intermediate edge to the node $v'$. The
communication from $u$ to $v'$ can be established along a shortest path
(in $\Z_{t-1}$). This shortest path can be computed locally, since every
node knows the complete virtual graph.
\end{compactdesc}
\ENDFOR
\STATE After all additional vertices have been generated, the nodes in
$I$, start initiating random walks of length $O(\log n)$ to distribute any
(new) vertices that exceed the treshold of $\NewLoad> 4 \zeta$.
These walks are performed sequentially in some arbitrary order. (Note that
$|I|\in O(1)$.)
\STATE Once these walks are complete, the coordinator is informed and
contacts the nodes $I'$ that simulate the next $1/\theta$ vertices of the
current virtual graph. When the adversary triggers the next step,
these nodes in turn locally generate their portion of $\Z(p_{i+1})$ and so forth. 
After $\theta n$ steps, Phase~1 is complete at all nodes.

\item[]
\item[] \textbf{Phase 2. Discard the old $p$-cycle:}
\STATE The coordinator sends another request to the set of nodes $I$ that
host the first $\lceil 1/\theta \rceil$ vertices in $S$.
\STATE This causes every node in $I$ to drop all edges of $\Z(p_i)$ and
stop simulating the corresponding vertices.
\STATE In the recovery of the next step, the coordinator forwards this
request to the next $1\lceil /\theta \rceil$ nodes and so forth.
After $\theta n$ steps, Phase~2 is complete and all nodes now
(exclusively) simulate the new virtual graph $\Z(p_{i+1})$.
\end{algorithmic}
\caption{{Procedure \stagInflate}}
\label{algo:stagInflate}
\end{algorithm*}

\begin{algorithm*}[h!]
\begin{algorithmic}[1]
\item[]
\STATE \textbf{Assumption:} Let $w$ be the \emph{coordinator} node that
maintains local counters of $\Spare$, $\Low$ and the network size (cf.\
Algorithm~\ref{algo:coordinator}). Moreover, the coordinator has computed
the prime number $p_s$ of the smaller $p$-cycle to which we deflate.
\item[]
\item[] \textbf{\boldmath Phase 1. Compute smaller $p$-cycle:}
\item[] Every node $u \in I$ does the following:
\STATE (Compute the new set of locally simulated virtual vertices
$\NewSim(u) \subset \Z(p_s)$.) 
Let $\alpha = \frac{p_{i}}{p_{s}}$.
For every $x \in \Sim(u)$ (i.e.\ $x\in \Z_{t-1}(p_i)$) we compute 
$y_x = \lfloor\frac{x}{\alpha} \rfloor$. \\
If there is no $x'<x$ such that $y_{x'} = y_x$, we add $y_x$ to $\NewSim(u)$.
This yields the (possibly empty) set \\
$\NewSim(u)=\{y_{x_1},\dots,y_{x_k}\}$, \\ 
where $x_1,\dots,x_k \in \Z_{t-1}(p_{i})$ are
a subset of the previously simulated vertices at $u$.
If $\NewSim(u)=\emptyset$, we mark $u$ as \emph{contending}.
For every vertex $y_{x_j}$, we set\\
$\cloud{y_{x_j}} = \{ m \colon  (m-1)\lfloor\alpha\rfloor \le y_{x_j}<m\lfloor\alpha\rfloor \}$.
\FOR{every $y_{x_j} \in \NewSim(u)$, $(1\le j\le k)$, } 
\item[] (Compute the new set of edges.)\\
\begin{compactdesc}
\item Cycle edges: Add an (intermediate) edge between $u$ and the nodes $v$ and $v'$
  that are going to simulate $y_{x_j}-1$ and $y_{x_j}+1$ by using the
  cycle edges of $\Z_{t-1}(p_i)$ in $G_t$.
\item Inverse edges: Add an(intermediate)  edge between $u$ and the node
  $v$ that is going to simulate $y_k^{-1}$; node $v$ is found by
  communicating along a shortest path in $\Z(p_i)$.
This shortest path can be computed locally, since every
node knows the complete virtual graph.

\end{compactdesc}
\ENDFOR
\STATE After all additional vertices have been generated, the
\emph{contending} nodes in $I$, start initiating random walks of length
$O(\log n)$ to find nodes that have $\NewLoad< 4\zeta$.
Note that, even though only nodes in $I$ have generated their part of the
new $p$-cycle, every node can locally compute its value of $\NewLoad$ upon
being hit by such a random walk and hence can generate such vertices on
the fly.
These walks are performed sequentially in some arbitrary order. (Note that
$|I|\in O(1)$.)
\STATE Once these walks are complete, the coordinator is informed and
contacts the nodes $I'$ that simulate the next $1/\theta$ vertices of the
current virtual graph.
When the adversary triggers the next step, these nodes in turn will
locally generate their portion of $\Z(p_{s})$ and so forth. 
After $\theta n$ steps, Phase~1 is complete at all nodes.
\item[]
\item[] \textbf{Phase 2. Discard the old $p$-cycle:}
\STATE The coordinator sends another request to the set of nodes $I$ that
host the first $\lceil 1/\theta \rceil$ vertices in $S$.
\STATE This causes every node in $I$ to drop all edges of $\Z(p_i)$ and
stop simulating the corresponding vertices.
\STATE In the recovery of the next step, the coordinator forwards this
request to the next $1\lceil /\theta \rceil$ nodes and so forth.
After $\theta n$ steps, Phase~2 is complete and all nodes now
(exclusively) simulate the new virtual graph $\Z(p_{s})$.
\end{algorithmic}
\caption{{Procedure \stagDeflate}}
\label{algo:stagDeflate}
\end{algorithm*}

%% file: Extensions.tex
\section{Extension: Handling Multiple Insertions and Deletions} \label{sec:multiplechanges}
\label{sec:extension}

Our framework can be extended to a model where the
adversary can insert or delete multiple nodes in each step, with
certain assumptions: 

\noindent\textbf{Insertions:} The adversary can insert or delete a set $N$
of up to $\eps n$ many nodes in each step, for some small $\eps>0$. We
restrict the adversary to attach only a constant number of nodes in $N$ to
any node---dropping this restriction will allow the adversary to place the
whole set $N$ at the same node $u$, causing significant congestion due to
$u$'s constant degree and our restriction of having messages of $O(\log
n)$ size.
Note that this might cause type-1 recovery to fail more frequently, since
the number of available spare vertices is depleted within a constant
number of insertion steps.
Nevertheless we can still handle such large-scale insertions via type-2
recovery by using Procedure $\inflate$.

\noindent\textbf{Deletions:} For deletions, we only allow the adversary to
delete nodes that leave the remainder graph connected, i.e., if the adversary removes
nodes $N$ at time $t$, $G_{t-1}\setminus N$ is still connected. Moreover,
for each deleted node there must remain at least one neighbor in the set
$G_{t-1}\setminus N$.
As in the case of insertions, such large-scale deletions might require
Procedure $\deflate$ to be invoked every constant number of steps.

\begin{corollary}[Multiple Insertions/Deletions]
Suppose that the adversary can insert or delete $\le \eps n$ nodes, for some
small $\eps > 0$ in every step adhering to the following conditions:
In case of insertions, the adversary attaches $O(1)$ nodes to any
existing node in the network.
In case of deletions, the remaining graph is connected and, for each
deleted node $u$, some neighbor of $u$ is not deleted.
There exists a distributed algorithm that requires $O(n\log^2 n)$ messages
and $O(\log^3 n)$ rounds (w.h.p.) for recovery in every step.
\end{corollary}


%% file: IPDPS13-DEX.bbl
\begin{thebibliography}{10}

\bibitem{DBLP:books/wi/AlonS92}
Noga Alon and Joel Spencer.
\newblock {\em The Probabilistic Method}.
\newblock Wiley, 1992.

\bibitem{AspnesW09}
James Aspnes and Udi Wieder.
\newblock The expansion and mixing time of skip graphs with applications.
\newblock {\em Distributed Computing}, 21(6):385--393, 2009.

\bibitem{APRU2012:Towards}
John Augustine, Gopal Pandurangan, Peter Robinson, and Eli Upfal.
\newblock Towards robust and efficient computation in dynamic peer-to-peer
  networks.
\newblock In {\em SODA}, 2012.

\bibitem{Bertrand1850}
J.~Bertrand.
\newblock Mémoire sur le nombre de valeurs que peut prendre une fonction quand
  on y permute les lettres qu'elle renferme.
\newblock {\em J. l'École Roy. Polytech. 17}, pages 123--140, 1845.

\bibitem{chungbook}
Fan Chung.
\newblock {\em Spectral Graph Theory}.
\newblock AMS, 1997.

\bibitem{CooperFlipMarkovChainPODC09-abbrv}
Colin Cooper, Martin Dyer, and Andrew~J. Handley.
\newblock The flip markov chain and a randomising p2p protocol.
\newblock In {\em PODC}. ACM, 2009.

\bibitem{drw1}
Atish {Das Sarma}, Danupon Nanongkai, and Gopal Pandurangan.
\newblock Fast distributed random walks.
\newblock In {\em PODC}, pages 161--170, 2009.

\bibitem{DolevSpanders2010}
Shlomi Dolev and Nir Tzachar.
\newblock Spanders: distributed spanning expanders.
\newblock In {\em SAC}, pages 1309--1314, 2010.

\bibitem{G1998:Chernoff}
David Gillman.
\newblock A {C}hernoff bound for random walks on expander graphs.
\newblock {\em SIAM J. Comput.}, 27(4):1203--1220, 1998.

\bibitem{mihail-p2p2006}
C.~Gkantsidis, M.~Mihail, and A.~Saberi.
\newblock Random walks in peer-to-peer networks: Algorithms and evaluation.
\newblock {\em Performance Evaluation}, 63(3):241--263, 2006.

\bibitem{GurevichK10}
Maxim Gurevich and Idit Keidar.
\newblock Correctness of gossip-based membership under message loss.
\newblock {\em SIAM J. Comput.}, 39(8):3830--3859, 2010.

\bibitem{HayesFG-DCJournal-springerlink}
Thomas Hayes, Jared Saia, and Amitabh Trehan.
\newblock The forgiving graph: a distributed data structure for low stretch
  under adversarial attack.
\newblock {\em Distributed Computing}, pages 1--18.
\newblock 10.1007/s00446-012-0160-1.

\bibitem{HayesPODC08-abbrv}
Tom Hayes, Navin Rustagi, Jared Saia, and Amitabh Trehan.
\newblock The forgiving tree: a self-healing distributed data structure.
\newblock In {\em PODC '08}. ACM, 2008.

\bibitem{Wigderson-exsurvey}
Shlomo Hoory, Nathan Linial, and Avi Wigderson.
\newblock {Expander graphs and their applications}.
\newblock {\em Bulletin of the AMS}, 43(04):439--562, 2006.

\bibitem{JacobSS-Skip09-abbrv}
Riko Jacob, Andrea Richa, Christian Scheideler, Stefan Schmid, and Hanjo
  T\"{a}ubig.
\newblock A distributed polylogarithmic time algorithm for self-stabilizing
  skip graphs.
\newblock In {\em PODC '09}. ACM, 2009.

\bibitem{KingFOCS06-abbrv}
Valerie King, Jared Saia, Vishal Sanwalani, and Erik Vee.
\newblock Towards secure and scalable computation in peer-to-peer networks.
\newblock In {\em FOCS}, 2006.

\bibitem{Kuhn2005-Repairing-abbrv}
Fabian Kuhn, Stefan Schmid, and Roger Wattenhofer.
\newblock Towards worst-case churn resistant peer-to-peer systems.
\newblock {\em Distributed Computing}, 22(4):249--267, 2010.

\bibitem{LS03-abbrv}
C.~Law and K.-Y. Siu.
\newblock Distributed construction of random expander networks.
\newblock In {\em INFOCOM 2003}, volume~3, 2003.

\bibitem{L1994:Discrete}
Alexander Lubotzky.
\newblock {\em Discrete groups, expanding graphs and invariant measures, vol
  125, Progress in Mathematics}.
\newblock Birkh{\"a}user, 1994.

\bibitem{MelamedK08}
Roie Melamed and Idit Keidar.
\newblock Araneola: A scalable reliable multicast system for dynamic
  environments.
\newblock {\em J. Parallel Distrib. Comput.}, 68(12):1539--1560, 2008.

\bibitem{MU2005:Probability}
Michael Mitzenmacher and Eli Upfal.
\newblock {\em Probability and {C}omputing}.
\newblock Cambridge University Press, 2005.

\bibitem{DBLP:journals/talg/NaorW07}
Moni Naor and Udi Wieder.
\newblock Novel architectures for p2p applications: The continuous-discrete
  approach.
\newblock {\em ACM Transactions on Algorithms}, 3(3), 2007.

\bibitem{PRU01}
Gopal Pandurangan, Prabhakar Raghavan, and Eli Upfal.
\newblock Building low-diameter {P2P} networks.
\newblock In {\em FOCS}, pages 492--499, 2001.

\bibitem{PanduranganPODC11-abbrv}
Gopal Pandurangan and Amitabh Trehan.
\newblock Xheal: localized self-healing using expanders.
\newblock In {\em PODC '11}. ACM, 2011.

\bibitem{pelegDCbook}
David Peleg.
\newblock {\em Distributed Computing: A Locality Sensitive Approach}.
\newblock SIAM, 2000.

\bibitem{ReiterDistNetwork-SRDS2005-abbrv}
M.K. Reiter, A.~Samar, and C.~Wang.
\newblock Distributed construction of a fault-tolerant network from a tree.
\newblock In {\em SRDS 2005}, 2005.

\bibitem{SaiaTrehanIPDPS08-abbrv}
Jared Saia and Amitabh Trehan.
\newblock Picking up the pieces: Self-healing in reconfigurable networks.
\newblock In {\em IPDPS}, 2008.

\bibitem{S1998:Universal-abbrv}
Christian Scheideler.
\newblock {\em Universal Routing Strategies for Interconnection Networks},
  volume 1390 of {\em LNCS}.
\newblock Springer.

\bibitem{Amitabh-2010-PhdThesis}
Amitabh Trehan.
\newblock {\em {Algorithms for self-healing networks}}.
\newblock Dissertation, {University of New Mexico}, 2010.

\end{thebibliography}
